\documentclass[journal,twoside,web,letter]{ieeecolor_noruler}
\usepackage{generic}
\usepackage{cite}
\usepackage{amsmath,amssymb,amsfonts}
\usepackage{graphicx}
\usepackage{xcolor}
\usepackage{algorithm,algorithmic}
\usepackage{hyperref}
\hypersetup{hidelinks=true}
\usepackage{textcomp}
\usepackage{wrapfig}
\usepackage{mathrsfs}

\newtheorem{corol}{Corollary}
\newtheorem{definition}{Definition}

\newtheorem{lemma}{Lemma}
\newtheorem{remark}{Remark}
\newtheorem{thm}{Theorem}

\DeclareMathOperator{\rank}{rank}

\DeclareMathOperator{\row}{row}

\def\BibTeX{{\rm B\kern-.05em{\sc i\kern-.025em b}\kern-.08em
		T\kern-.1667em\lower.7ex\hbox{E}\kern-.125emX}}
\begin{document}
	\def\ZZ{{\mathbb Z}}
	\def\RR{{\Bbb R}}
	\def\NN{{\mathbb N}}
	\def\CC{{\mathbb C}}
	
	\title{A Minimum-Order Functional Observer Beyond the Darouach and Luenberger Constructions}
	\author{ Tyrone Fernando
		\thanks{T.~Fernando is with the Department of Electrical, Electronic and Computer Engineering, University of Western Australia (UWA), 35 Stirling Highway, Crawley, WA 6009, Australia. (email:  tyrone.fernando@uwa.edu.au)}	
	}
	
	\maketitle
\begin{abstract}
	This paper develops a general minimum-order functional observer
	framework that contains the classical Darouach functional observer and
	the reduced-order Luenberger observer as special cases. The required
	functional dimension is determined from the triple $(A,C,L)$ through a
	set of functional observability indices. The minimum dimension permitted
	by the algebraic observer-existence condition is characterized, and the
	complete family of augmentations attaining this dimension is determined.
	Necessary-and-sufficient spectral feasibility conditions are then
	established for this dimension to be the minimum achievable observer
	order. The resulting existence conditions reduce exactly to the
	Darouach conditions when no augmentation is required and to the
	classical observability condition in the reduced-order Luenberger case.
	More generally, the framework admits minimum-order functional observers
	of intermediate order, which may exist even when neither classical
	observer is available.
\end{abstract}

	
\section{Introduction}
\label{sec:introduction}

Functional observers arise when only prescribed linear functions of the
state, rather than the entire state vector, need to be estimated.
Consider the continuous-time linear time-invariant system
\begin{equation}
	\dot x(t)=Ax(t)+Bu(t),
	\quad
	y(t)=Cx(t),
	\quad
	z(t)=Lx(t),
	\label{eq:1}
\end{equation}
where $x(t)\in\mathbb R^n$ is the state, $u(t)\in\mathbb R^m$ is the input,
$y(t)\in\mathbb R^p$ is the measured output, and $z(t)\in\mathbb R^r$ is the
functional to be estimated, with
$A\in\mathbb R^{n\times n}$,
$B\in\mathbb R^{n\times m}$,
$C\in\mathbb R^{p\times n}$, and
$L\in\mathbb R^{r\times n}$. When
$r\ll n$, reconstructing the complete state solely to recover $z$ may
introduce unnecessary observer dynamics, motivating the search for the
minimum observer order required to estimate $Lx(t)$.

The reduced-order Luenberger observer provides a minimum-order structure
for reconstructing state directions not directly available from the
measured output. Luenberger also introduced functional observers for
estimating prescribed linear functions without reconstructing the
entire state~\cite{Luenberger1966}, although general existence conditions were
not established. This initiated extensive research on single- and
multi-functional observer design; see
\cite{FortmannWilliamson1972}-\cite{ref9n} and the references therein.
Darouach~\cite{22new} later established necessary-and-sufficient
conditions for the existence and design of an order-$r$ functional
observer, where $r$ is the number of independent functional directions
to be estimated.
More recently, functional-observer theory has been extended
to large-scale and interconnected systems~\cite{ref14nb}-\cite{ref14nd},
networked systems~\cite{ref14n}-\cite{mont}, sampled-data
systems~\cite{ref12n} and nonlinear systems \cite{nonl1}-\cite{nonl3}. The Darouach functional observer and the
reduced-order Luenberger observer represent two important classical
cases of the minimum-order problem considered here.

These cases do not exhaust the possible functional-observer structures.
The prescribed functional matrix $L$ may fail the Darouach
observer-existence conditions, while the system may be unobservable and
hence admit no reduced-order Luenberger observer, even though the triple $(A,C,L)$
is functional observable. Failure of an order-$r$ Darouach observer does not imply that a
minimum-order functional observer does not exist, nor that all $n-p$
unmeasured state directions must be reconstructed; the minimum observer
order may lie strictly between $r$ and $n-p$. Thus, the minimum order is not restricted to either the
number of prescribed functionals or the number of unmeasured state
directions, but is determined more generally by the functional
observability structure of $(A,C,L)$, subject to spectral feasibility.

This paper develops a general minimum-order functional observer
framework containing the Darouach and reduced-order Luenberger
observers as special cases while admitting this intermediate regime.
The prescribed functional matrix is augmented as
\begin{equation}
	\mathcal L
	=
	\begin{pmatrix}
		L\\
		R
	\end{pmatrix},
	\nonumber
\end{equation}
where the required number of additional functional directions is
determined from $(A,C,L)$ through functional observability indices and
their choice is made from the complete family of admissible
augmentations.

The observer-existence conditions are separated into an algebraic
condition $(A)$ and a spectral condition $(S)$. The minimum functional
dimension satisfying condition $(A)$ is
\begin{equation}
	q_0
	=
	\sum_{j=1}^{r}\eta_j,
	\qquad
	r\leq q_0\leq n-p,
	\nonumber
\end{equation}
where $\eta_1,\ldots,\eta_r$ are the functional observability indices associated with the rows of $L$.
Since condition $(A)$ is necessary for observer existence, no
functional observer can have order smaller than $q_0$.

An augmentation attaining $q_0$ need not be unique. Hence,
failure of condition $(S)$ for one minimum condition-$(A)$ augmentation
does not preclude another augmentation of the same dimension from
satisfying both conditions. The complete family of minimum
condition-$(A)$ augmentations is therefore characterized, together with
necessary-and-sufficient conditions for the existence of a member
satisfying condition $(S)$. Consequently,
\begin{equation}
	\nu_{\min}=q_0
	\nonumber
\end{equation}
if and only if at least one member of this family satisfies condition
$(S)$.

The classical cases follow directly. When the Darouach
observer-existence conditions hold, no augmentation is required,
$\eta_j=1$ for $j=1,\ldots,r$, and
\begin{equation}
	\nu_{\min}=q_0=r,
	\nonumber
\end{equation}
with the general existence conditions reducing exactly to the
necessary-and-sufficient Darouach conditions. At the other
specialization, when the prescribed functional directions represent
all state directions not directly available from the measured output,
\begin{equation}
	\nu_{\min}=q_0=n-p,
	\nonumber
\end{equation}
and the general existence condition reduces to the classical
observability condition for the reduced-order Luenberger observer.

The principal extension is the intermediate regime
\begin{equation}
	r<q_0<n-p.
	\nonumber
\end{equation}
Here an order-$r$ Darouach observer does not exist, while reconstruction
of all $n-p$ unmeasured state directions is unnecessary. If spectral
feasibility holds at $q_0$, then
$\nu_{\min}=q_0$, yielding a minimum-order functional observer
that may exist even when the system is unobservable and neither
classical construction is available.

The main contributions are:
\begin{enumerate}
	\item The minimum functional dimension permitted by condition $(A)$
	is characterized by the functional observability indices as
	\(
	q_0=\sum_{j=1}^{r}\eta_j,
	\)
	with
	\(
	r\leq q_0\leq n-p.
	\)
	
	\item The complete family of augmentations attaining
	$q_0$ and satisfying condition $(A)$ is characterized, and
	necessary-and-sufficient conditions are established for
	\(
	\nu_{\min}=q_0.
	\)
	
	\item The Darouach and reduced-order Luenberger observer-existence
	conditions are recovered as special cases, while intermediate-order
	minimum functional observers are admitted when neither classical
	construction is applicable.
	
\end{enumerate}

The paper concludes with an example in which the Darouach conditions
fail and the system is unobservable, yet an intermediate minimum-order
functional observer exists. The example also shows why the complete
minimum condition-$(A)$ family must be considered: one minimum
augmentation fails condition $(S)$, whereas another of the same
dimension satisfies it.

\section{Preliminaries and Problem Formulation}
\label{sec:problem}

Consider the system and prescribed functional in \eqref{eq:1}, with
input $u(t)$ and measured output $y(t)$ available. Without loss of generality,
$C$ and $L$ are assumed to have full row rank: dependent measurements
provide no additional information, while dependent prescribed
functionals can be recovered from estimates of a maximal independent
subset. We further assume
\begin{equation}
	\rank
	\begin{pmatrix}
		L\\
		C
	\end{pmatrix}
	=r+p,
	\nonumber
\end{equation}
since any component of $L$ in $\row(C)$ is directly available from the
measured output and requires no observer dynamics.

Throughout, the triple $(A,C,L)$ is assumed functional observable, so that
$z(t)=Lx(t)$ is reconstructible from the available input and output. The objective is to determine a minimum-order functional observer for
estimating $z$ and to establish necessary-and-sufficient conditions for
its existence.

\subsection{Functional Observability}
\label{subsec:functional-observability}

Functional observability admits two equivalent characterizations.

\begin{definition}[Functional Observability: Reconstruction Form]
	\label{def:FO-reconstruction}
	As defined in \cite{ref3n}, $(A,C,L)$ is functionally observable if
	there exists a finite $T>0$ such that $z(0)=Lx(0)$ can be uniquely
	determined from the output $y(t)$ and known input $u(t)$ over
	$0\leq t\leq T$.
\end{definition}

An equivalent augmentation-based characterization is as follows.

\begin{definition}[Functional Observability: Augmentation Form]
	\label{def:FO-augmentation}
	As defined in \cite{ref10}, $(A,C,L)$ is functionally observable if
	there exists a finite-dimensional augmentation
	$R\in\mathbb R^{q_R\times n}$ such that
	\(
	\mathcal L=
	\begin{pmatrix}
		L\\
		R
	\end{pmatrix}
	\)
	satisfies the algebraic condition
	\begin{equation}
		\rank
		\begin{pmatrix}
			\mathcal LA\\
			\mathcal L\\
			CA\\
			C
		\end{pmatrix}
		=
		\rank
		\begin{pmatrix}
			\mathcal L\\
			CA\\
			C
		\end{pmatrix},
		\tag{A}
		\label{eq:conditionA}
	\end{equation}
	and the spectral condition
	\begin{equation}
		\rank
		\begin{pmatrix}
			\lambda\mathcal L-\mathcal LA\\
			CA\\
			C
		\end{pmatrix}
		=
		\rank
		\begin{pmatrix}
			\mathcal L\\
			CA\\
			C
		\end{pmatrix},
		\quad
		\forall \lambda\in\mathbb C.
		\tag{S}
		\label{eq:conditionS}
	\end{equation}
\end{definition}

\begin{remark}
	Condition $(S)$, as stated for all $\lambda\in\mathbb C$, corresponds to
	arbitrary observer-pole assignment. If only asymptotic convergence
	is required, it suffices that
	\(
	\rank
	\begin{pmatrix}
		\lambda\mathcal L-\mathcal LA\\
		CA\\
		C
	\end{pmatrix}
	=
	\rank
	\begin{pmatrix}
		\mathcal L\\
		CA\\
		C
	\end{pmatrix},
	\qquad
	\Re(\lambda)\geq0,
	\)
	so that the observer poles can be assigned to the open left-half
	plane.
\end{remark}

As established in \cite{ref3n,ref4n},
Definitions~\ref{def:FO-reconstruction} and
\ref{def:FO-augmentation} are equivalent and yield
\begin{equation}
	\row(\mathcal O_L)\subseteq\row(\mathcal O_C),
	\,\,
	\mathcal O_C=
	\begin{pmatrix}
		C\\ CA\\ \vdots\\ CA^{n-1}
	\end{pmatrix},
	\,\,
	\mathcal O_L=
	\begin{pmatrix}
		L\\ LA\\ \vdots\\ LA^{n-1}
	\end{pmatrix}.
	\nonumber
\end{equation}
Equivalently,
\begin{equation}
	\rank
	\begin{pmatrix}
		\mathcal O_C\\
		\mathcal O_L
	\end{pmatrix}
	=
	\rank(\mathcal O_C),
	\label{eq:functional-observability-rank}
\end{equation}
where $\row(M)$ denotes the row space of a matrix $M$.
For $L=I_n$, this reduces to $\rank(\mathcal O_C)=n$, recovering
classical state observability. Also, condition $(A)$ is equivalently
\begin{equation}
	\row(\mathcal LA)\subseteq\row(\mathcal L,C,CA).
	\nonumber
\end{equation}

Since the observer order equals the number of independent functional
directions represented by $\mathcal L$, the minimum-order problem is 
\begin{equation}
	\nu_{\min}
	:=
	\min_{\mathcal L}\rank\mathcal L,
	\qquad
	\row(L)\subseteq\row(\mathcal L),
	\nonumber
\end{equation}
subject to conditions \eqref{eq:conditionA} and
\eqref{eq:conditionS}. The analysis proceeds in three stages:
\begin{enumerate}
	\item[I] determine the smallest functional dimension permitted by
	$(A)$ and construct an augmentation attaining it;
	\item[II] characterize all minimum-dimensional augmentations
	satisfying $(A)$;
	\item[III] test $(S)$ over this complete family.
\end{enumerate}
This separation is necessary because the minimum dimension imposed by
$(A)$ is unique, whereas the augmentation attaining it need not be.

	\section{Ordered Functional-Observability Construction}
	\label{sec:ordered-FO-construction}
	
	The functional-observability condition \cite{ref10n},
	\begin{equation}
		\row(\mathcal O_L)
		\subseteq
		\row(\mathcal O_C) \nonumber 
	\end{equation}
	establishes that the functional information generated by the rows of
	$L$ is contained in the information generated by the measured output.
	For observer construction, however, it is necessary to identify how
	this information is generated and which additional functional
	directions may be used to form an augmented functional matrix
	satisfying conditions (A) and (S).
	
	To this end, we introduce an ordered rank-retention construction based
	jointly on the rows of $C$ and $L$ and their successive $A$-iterates.
	The construction produces a finite set of independent directions from
	which the functional augmentations considered later will be formed.

	\subsection{Global examination order and rank-retention rule}
	\label{subsec:global-rank-retention}
	
	Write
	\begin{equation}
		C=(C_1;\ldots;C_p),
		\quad
		L=(L_1;\ldots;L_r), \nonumber 
	\end{equation}
	and define the $p+r$ {\it generating rows} by
	\begin{equation}
		g_1,\ldots,g_{p+r}
		:=
		C_1,\ldots,C_p,L_1,\ldots,L_r.
	\label{eq:generating-rows}
	\end{equation}
	For each generator $g_i$, consider the $A$-generated sequence
	$$
	g_i,\quad g_iA,\quad g_iA^2,\quad\ldots.
	$$
	The rows
	$$
	g_iA^q,
	\quad
	i=1,\ldots,p+r,
	\quad
	q=0,1,2,\ldots,
	$$
	are referred to as the \emph{candidate rows}.
	
	The candidate rows are examined in the following \emph{global order}:
	by increasing powers of $A$ and, within each power, in the fixed
	generator order
	\eqref{eq:generating-rows}:
	\begin{equation}
		g_1,\ldots,g_{p+r};\;
		g_1A,\ldots,g_{p+r}A;\;
		g_1A^2,\ldots,g_{p+r}A^2;\;
		\cdots. 
		\label{eq:compact-global-order}
	\end{equation}
	
	The \emph{global rank-retention procedure} is defined as follows. Let $v_k$ denote the $k^{\mathrm{th}}$ candidate row in this global order,
	and let $\mathcal M_{k-1}$ denote the matrix of rows retained after
	examining the first $k-1$ candidate rows, with $\mathcal M_0$ empty.
	The candidate $v_k$ is retained, that is, appended as a new row of
	$\mathcal M_{k-1}$, if and only if it increases the rank:
	\begin{equation}
			v_k\text{ is retained}
			\iff
			\operatorname{rank}
			\begin{pmatrix}
				\mathcal M_{k-1}\\
				v_k
			\end{pmatrix}
			>
			\operatorname{rank}(\mathcal M_{k-1}). \nonumber 
	\end{equation}
	Equivalently,
	$$
	\mathcal M_k
	=
	\begin{cases}
		\begin{pmatrix}
			\mathcal M_{k-1}\\
			v_k
		\end{pmatrix},
		&
		\text{if $v_k$ is retained},\\[4mm]
		\mathcal M_{k-1},
		&
		\text{otherwise}.
	\end{cases}
	$$
	Thus $\mathcal M_k$ contains exactly the candidate rows retained after
	the first $k$ examinations. Since a candidate row is retained only
	when it increases the rank, the rows of every $\mathcal M_k$ are
	linearly independent.
	
	Although the candidate sequence in the global order 
	\eqref{eq:compact-global-order} is written indefinitely, only finitely
	many powers need to be examined for each generator $g_i$. By the
	Cayley--Hamilton theorem,
	$$
	A^n
	\in
	\operatorname{span}\{I,A,\ldots,A^{n-1}\},
	$$
	and hence, for every generator $g_i$,
	$$
	g_iA^n
	\in
	\operatorname{span}
	\{g_i,g_iA,\ldots,g_iA^{n-1}\}.
	$$
	Therefore every $A$-generated sequence encounters a dependent row no
	later than power $n$. By Cayley--Hamilton, only powers up to $A^{n-1}$ need be considered.
	Moreover, by the standard observability-index construction, dependence
	persists along each $A$-generated sequence: if $g_iA^q$ fails to
	increase the rank when examined, then so does $g_iA^{q+\ell}$ for every
	$\ell\geq1$.

	\begin{lemma}[Persistence of dependence]
		\label{lem:persistence-of-dependence}
		In the global rank-selection procedure
		\eqref{eq:compact-global-order}, if \(g_iA^q\) is dependent on the
		previously retained rows when it is examined, then every subsequent
		row
		\(
		g_iA^{q+\ell},
		\,\, \ell\geq1,
		\)
		is also dependent on the previously retained rows when it is
		examined.
	\end{lemma}

	\begin{proof}
		For each step $k$ in the global order, let
		\(
		S_k:=\row(\mathcal M_k)
		\)
		denote the row space retained after the first $k$
		examinations. Since retained rows are only added and never removed,
		\(
		S_0\subseteq S_1\subseteq S_2\subseteq\cdots.
		\)
		
		We will use the following observation: every candidate row already
		examined belongs to each subsequent retained-row space. Indeed, if
		$v_j$ is retained, then $v_j\in S_j$; if it is not retained, then
		$v_j\in S_{j-1}$. Since the spaces $S_k$ are nested, in either case
		\[
		v_j\in S_k,\qquad k\geq j.
		\]
		We first show that dependence of $g_iA^q$ implies dependence of the
		next row $g_iA^{q+1}$.
		
		Suppose $g_iA^q$ is examined at step $k$ and is dependent. Then
		\[
		g_iA^q\in S_{k-1}.
		\]
		Hence it can be written as a linear combination of rows retained
		before step $k$:
		\[
		g_iA^q
		=
		\sum_\nu \alpha_\nu g_{j_\nu}A^{s_\nu}.
		\]
		Because each row on the right was examined before $g_iA^q$, the
		global order implies that, for every $\nu$, either
	$
		s_\nu<q,
	$
		or
	$
		s_\nu=q
		\quad\text{and}\quad
		j_\nu<i.
	$
		Multiplying the dependence relation by $A$ gives
		\[
		g_iA^{q+1}
		=
		\sum_\nu \alpha_\nu g_{j_\nu}A^{s_\nu+1}.
		\]
		Now let $k'>k$ be the step at which $g_iA^{q+1}$ is examined.
		We claim that every row on the right-hand side has already been
		examined before step $k'$. Indeed, if $s_\nu<q$, then
		\[
		s_\nu+1\leq q<q+1,
		\]
		so $g_{j_\nu}A^{s_\nu+1}$ occurs in an earlier power block than
		$g_iA^{q+1}$. If $s_\nu=q$, then $j_\nu<i$, so
		$g_{j_\nu}A^{q+1}$ occurs earlier than $g_iA^{q+1}$ within the
		same power-$(q+1)$ block.
		
		Thus every row on the right-hand side has been examined before
		$g_iA^{q+1}$. By the observation above, each such row therefore
		belongs to $S_{k'-1}$. Consequently,
		\[
		g_iA^{q+1}\in S_{k'-1}.
		\]
		Hence $g_iA^{q+1}$ is dependent on the rows retained before its
		own examination.
		
		The same argument can now be applied again to $g_iA^{q+1}$, then
		to $g_iA^{q+2}$, and so on. Therefore, for every $\ell\geq1$,
		\[
		g_iA^{q+\ell}
		\]
		is dependent on the previously retained rows when it is examined.
	\end{proof}

Consequently, once the first dependent row in an $A$-generated
sequence is encountered, no higher power from that sequence can
contribute a new independent direction. The sequence may therefore be
terminated at that point. Since each generator encounters a dependent
row no later than power $n$, the ordered rank-retention procedure
terminates after finitely many examinations.

Let $\mathcal M$ denote the final matrix formed by the rows retained
by this procedure, listed in the order in which they are retained.
Thus, $\mathcal M$ contains exactly the independent candidate rows
selected by the ordered rank-retention rule. By construction, its rows
are linearly independent, and therefore
$$
\operatorname{rank}(\mathcal M)
=
\text{number of rows of }\mathcal M
\leq n.
$$

\subsection{Functional-observability indices}
\label{subsec:functional-observability-indices}

The functional-observability indices are read directly from the
retained matrix $\mathcal M$. Since
$$
\operatorname{rank}
\begin{pmatrix}
	C\\
	L
\end{pmatrix}
=p+r,
$$
all generating rows $g_1,\ldots,g_{p+r}$ are retained. Hence, for
each generator $g_i$, define
\begin{equation}
		\alpha_i
		=
		1+
		\max\left\{
		k\geq0:
		g_iA^k
		\text{ is a row of }\mathcal M
		\right\},
		\,\,
		i=1,\ldots,p+r. \nonumber 
\end{equation}

Thus the highest retained power associated with $g_i$ is
$$
g_iA^{\alpha_i-1},
$$
whereas
$$
g_iA^{\alpha_i}
$$
is the first dependent row of the corresponding $A$-generated
sequence. By Lemma~\ref{lem:persistence-of-dependence}, all subsequent
powers generated by $g_i$ are also dependent.

The integers
\begin{equation}
		\alpha_1,\ldots,\alpha_{p+r} \nonumber 
\end{equation}
are called the {\it ordered functional-observability indices}.

No regrouping of the rows of $\mathcal M$ is required to determine
these indices. For each generator, $\alpha_i$ is obtained by
identifying its highest retained power in $\mathcal M$ and adding one.
The rows themselves remain in the global order induced by
\eqref{eq:compact-global-order}.

For subsequent use, denote the functional-observability indices
associated specifically with the rows of $L$ by
\begin{equation}
		\eta_j:=\alpha_{p+j},
		\quad
		j=1,\ldots,r. \nonumber 
\end{equation}
Thus $\eta_j$ is the functional-observability index associated with
the generator
$$
g_{p+j}=L_j.
$$

\subsection{First-dependence relations}
\label{subsec:first-dependence-relations}

For each $i=1,\ldots,p+r$, the functional-observability index
$\alpha_i$ identifies the first dependent row of the generator $g_i$,
namely
$$
g_iA^{\alpha_i}.
$$
Let $\mathcal M_i^-$ denote the matrix formed by all rows retained
strictly before $g_iA^{\alpha_i}$ is reached in the prescribed global
order. By the rank-retention rule,
\begin{equation}
	g_iA^{\alpha_i}
	\in
	\operatorname{row}(\mathcal M_i^-). \nonumber 
\end{equation}
Since the rows of $\mathcal M_i^-$ are linearly independent, there
exists a unique coefficient row vector $\phi_i$ such that
\begin{equation}
		g_iA^{\alpha_i}
		=
		\phi_i\mathcal M_i^-,
		\qquad
		i=1,\ldots,p+r. 
	\label{eq:first-dependence-representation}
\end{equation}
The predecessor matrix $\mathcal M_i^-$ may contain retained rows
generated by both $C$ and $L$. Every row of $\mathcal M_i^-$ has the
form
$$
g_jA^q
$$
for some $j=1,\ldots,p+r$ and some nonnegative integer $q$.

The prescribed global examination order implies that such a row can
precede $g_iA^{\alpha_i}$ only if
$$
q<\alpha_i,
$$
or if
$$
q=\alpha_i
\quad\text{and}\quad
j<i.
$$
Thus, at powers strictly below $\alpha_i$, retained rows from any
generator may occur in $\mathcal M_i^-$, whereas at the terminal
power $\alpha_i$, only retained rows associated with the preceding
generators $g_1,\ldots,g_{i-1}$ can occur.

Equation \eqref{eq:first-dependence-representation} therefore
expresses each first dependent row directly in terms of the rows
actually retained before it, while preserving the ordering induced
by the global rank-retention procedure.

The collection of these dependence relations preserves the joint
$C$- and $L$-generated structure of the ordered construction. This
structure will now be used to construct augmented functional matrices
satisfying condition (A).

\section{Construction of Admissible Functional Augmentations}
\label{sec:augmentation-construction}

The ordered rank-retention construction provides a finite set of
independent directions generated jointly by the rows of $C$ and $L$
and their successive $A$-iterates. We now use the associated
first-dependence relations to construct augmented functional matrices
\begin{equation}
	\mathcal L
	=
	\begin{pmatrix}
		L\\
		R
	\end{pmatrix} \nonumber 
\end{equation}
that satisfy the observer-existence conditions (A) and (S).

Condition (A) requires algebraic closure under \(A\), modulo the measured directions supplied by \(C\) and \(CA\), whereas condition (S) imposes the corresponding spectral rank requirement. We first construct an augmentation satisfying (A).

\subsection{Augmentation required by condition (A)}
\label{subsec:augmentation-condition-A}

Condition $(A)$ is
\begin{equation}
	\rank
	\begin{pmatrix}
		\mathcal LA\\
		\mathcal L\\
		CA\\
		C
	\end{pmatrix}
	=
	\rank
	\begin{pmatrix}
		\mathcal L\\
		CA\\
		C
	\end{pmatrix}. \nonumber 
\end{equation}
Equivalently,
\begin{equation}
		\row(\mathcal LA)
		\subseteq
		\row
		\begin{pmatrix}
			\mathcal L\\
			CA\\
			C
		\end{pmatrix}. \nonumber 
\end{equation}

Thus condition $(A)$ requires the $A$-image of every row of
$\mathcal L$ to belong to the row space generated by $\mathcal L$, $C$,
and $CA$.

For the unaugmented choice $\mathcal L=L$, condition $(A)$ holds if and
only if
\begin{equation}
	\row(LA)
	\subseteq
	\row
	\begin{pmatrix}
		L\\
		CA\\
		C
	\end{pmatrix}. \nonumber 
\end{equation}
If this inclusion fails, additional functional directions are required.
The first-dependence relations obtained in
Section~\ref{sec:ordered-FO-construction} provide the information needed
to construct these directions.

\subsection{Construction of the augmentation for condition (A)}
\label{subsec:augmentation-construction-A}

For the functional generator $g_{p+j}=L_j, j\in\{1,\dots,r\}$, let

$$
\eta_j=\alpha_{p+j}.
$$
Then
$$
L_j,\;
L_jA,\;
\ldots,\;
L_jA^{\eta_j-1}
$$
are retained, whereas $L_jA^{\eta_j}$ is the first dependent row.
Hence, by \eqref{eq:first-dependence-representation},
\begin{equation}
		L_jA^{\eta_j}
		=
		\phi_{p+j}\mathcal M_{p+j}^{-},
		\qquad
		j=1,\ldots,r. 	\label{eq:Lj-first-dependence}
\end{equation}
We refer to \eqref{eq:Lj-first-dependence} as the {\it first-dependence
relation} associated with $L_j$.

Partition $\mathcal M_{p+j}^{-}$ according to the powers of $A$ in
its retained rows:
\begin{equation}
		\mathcal M_{p+j}^{-}
		=
		\begin{pmatrix}
			\mathcal G_{p+j}^{(0)}A^0\\
			\mathcal G_{p+j}^{(1)}A^1\\
			\vdots\\
			\mathcal G_{p+j}^{(\eta_j)}A^{\eta_j}
		\end{pmatrix}. \nonumber 
\end{equation}
Here $\mathcal G_{p+j}^{(q)}$ consists of those generators $g_i$,
$i=1,\ldots,p+r$, for which the candidate row $g_iA^q$ was retained
before $L_jA^{\eta_j}$  (i.e., $g_{p+j}A^{\eta_j}$) was examined. If no such
candidate row exists for a given $q$, then
$\mathcal G_{p+j}^{(q)}$ is empty.

Partition the coefficient row $\phi_{p+j}$ conformably as

$$
\phi_{p+j}
=
\begin{pmatrix}
	\phi_{p+j}^{(0)}&
	\phi_{p+j}^{(1)}&
	\cdots&
	\phi_{p+j}^{(\eta_j)}
\end{pmatrix}.
$$
With this partition, the first-dependence relation
\eqref{eq:Lj-first-dependence} becomes
\begin{equation}
		L_jA^{\eta_j}
		-
		\sum_{q=0}^{\eta_j}
		\phi_{p+j}^{(q)}
		\mathcal G_{p+j}^{(q)}A^q
		=
		{\bf 0}. 	\label{eq:Lj-partitioned-dependence}
\end{equation}
If $\mathcal G_{p+j}^{(q)}$ is empty for some $q$, the corresponding
term in \eqref{eq:Lj-partitioned-dependence} is omitted.
Writing \eqref{eq:Lj-partitioned-dependence} explicitly in descending
powers of $A$ gives
\begin{IEEEeqnarray}{rcl}
	&&L_jA^{\eta_j}
	-
	\phi_{p+j}^{(\eta_j)}
	\mathcal G_{p+j}^{(\eta_j)}A^{\eta_j}
	-
	\phi_{p+j}^{(\eta_j-1)}
	\mathcal G_{p+j}^{(\eta_j-1)}A^{\eta_j-1}
	-\cdots \nonumber \\
	&&-
	\phi_{p+j}^{(2)}
	\mathcal G_{p+j}^{(2)}A^2
	-
	\phi_{p+j}^{(1)}
	\mathcal G_{p+j}^{(1)}A
	-
	\phi_{p+j}^{(0)}
	\mathcal G_{p+j}^{(0)} = {\bf 0}. \nonumber 
\end{IEEEeqnarray}
The first augmented functional row, $\ell_{j,\eta_j-1}$, is constructed
from the terms on the left-hand side of
\eqref{eq:Lj-partitioned-dependence} corresponding to
$q=2,\ldots,\eta_j$. Specifically, omit the terms corresponding to
$q=0,1$ and decrease each remaining power of $A$ by one. This gives
\begin{IEEEeqnarray}{rcl}
\ell_{j,\eta_j-1}
&=&
L_jA^{\eta_j-1}
-
\phi_{p+j}^{(\eta_j)}
\mathcal G_{p+j}^{(\eta_j)}A^{\eta_j-1}
-
\phi_{p+j}^{(\eta_j-1)}
\mathcal G_{p+j}^{(\eta_j-1)}A^{\eta_j-2} \nonumber \\
&&-\cdots-
\phi_{p+j}^{(3)}
\mathcal G_{p+j}^{(3)}A^2
-
\phi_{p+j}^{(2)}
\mathcal G_{p+j}^{(2)}A. \nonumber
\end{IEEEeqnarray}
The remaining augmented functional rows are obtained successively from
the preceding row. To obtain the next row, decrease each positive power
of $A$ by one and discard the last term. Thus
\begin{IEEEeqnarray}{rcl}
\ell_{j,\eta_j-2}
&=&
L_jA^{\eta_j-2}
-
\phi_{p+j}^{(\eta_j)}
\mathcal G_{p+j}^{(\eta_j)}A^{\eta_j-2}
-
\phi_{p+j}^{(\eta_j-1)}
\mathcal G_{p+j}^{(\eta_j-1)}A^{\eta_j-3} \nonumber \\
&&-\cdots-
\phi_{p+j}^{(4)}
\mathcal G_{p+j}^{(4)}A^2
-
\phi_{p+j}^{(3)}
\mathcal G_{p+j}^{(3)}A. \nonumber
\end{IEEEeqnarray}
Continuing in this way gives
$$
\ell_{j,2}
=
L_jA^2
-
\phi_{p+j}^{(\eta_j)}
\mathcal G_{p+j}^{(\eta_j)}A^2
-
\phi_{p+j}^{(\eta_j-1)}
\mathcal G_{p+j}^{(\eta_j-1)}A,
$$
$$
\ell_{j,1}
=
L_jA
-
\phi_{p+j}^{(\eta_j)}
\mathcal G_{p+j}^{(\eta_j)}A,
$$
and finally
$$
\ell_{j,0}=L_j.
$$

This construction can be written compactly as
\begin{equation}
		\ell_{j,t}
		=
		L_jA^t
		-
		\sum_{q=\eta_j-t+1}^{\eta_j}
		\phi_{p+j}^{(q)}
		\mathcal G_{p+j}^{(q)}
		A^{q-\eta_j+t}, 	\label{eq:functional-row-construction}
\end{equation}
for $t=1,\ldots,\eta_j-1$, with
$$
\ell_{j,0}=L_j.
$$
These functional rows satisfy a first-order recursion. For
$t=0,\ldots,\eta_j-2$,
\begin{equation}
	\ell_{j,t}A
	=
	\ell_{j,t+1}
	+
	\phi_{p+j}^{(\eta_j-t)}
	\mathcal G_{p+j}^{(\eta_j-t)}A.
	\label{eq:functional-chain-recursion}
\end{equation}

For the terminal row, first suppose that $\eta_j>1$. Setting
$t=\eta_j-1$ in \eqref{eq:functional-row-construction}, multiplying by
$A$, and using \eqref{eq:Lj-partitioned-dependence} gives
\begin{equation}
	\ell_{j,\eta_j-1}A
	=
	\phi_{p+j}^{(0)}
	\mathcal G_{p+j}^{(0)}
	+
	\phi_{p+j}^{(1)}
	\mathcal G_{p+j}^{(1)}A.
	\label{eq:functional-chain-terminal}
\end{equation}
When $\eta_j=1$, no augmented row is required and
$\ell_{j,0}=L_j$. In this case,
\eqref{eq:functional-chain-terminal} follows directly from
\eqref{eq:Lj-partitioned-dependence}. Hence
\eqref{eq:functional-chain-terminal} holds for every $\eta_j\geq1$.

For each functional generator $L_j$, define the corresponding
augmentation block by
\begin{equation}
	R_j
	=
	\begin{pmatrix}
		\ell_{j,\eta_j-1}\\
		\vdots\\
		\ell_{j,1}
	\end{pmatrix},
	\qquad
	j=1,\ldots,r.
	\nonumber
\end{equation}
Thus $R_j$ contains the $\eta_j-1$ additional functional rows
constructed from the retained $A$-generated sequence of $L_j$.
If $\eta_j=1$, no additional row is required and $R_j$ is empty.

Stacking these augmentation blocks gives
\begin{equation}
	R
	=
	\begin{pmatrix}
		R_1\\
		\vdots\\
		R_r
	\end{pmatrix}.
	\label{eq:augmentation-matrix-R}
\end{equation}
The augmented functional matrix is therefore
\begin{equation}
	\mathcal L
	=
	\begin{pmatrix}
		L\\
		R
	\end{pmatrix}.
	\label{eq:augmented-functional-matrix}
\end{equation}

Thus the ordered rank-retention procedure provides the augmentation
matrix $R$ explicitly. The following result proves that the resulting
augmented functional matrix $\mathcal L$ satisfies condition $(A)$,
while the subsequent result proves that the number of additional rows
in $R$ is minimum.

\subsection{Verification of condition (A)}
\label{subsec:verification-condition-A}

The preceding construction converts each higher-order
first-dependence relation into a sequence of first-order relations.
The following result establishes that the resulting augmented
functional matrix satisfies condition $(A)$.

\begin{thm}
	\label{prop:augmentation-satisfies-A}
	Let $\mathcal L$ be the augmented functional matrix constructed in
	\eqref{eq:augmented-functional-matrix}, with $R$ defined by
	\eqref{eq:augmentation-matrix-R}. Then
	\begin{equation}
			\row(\mathcal LA)
			\subseteq
			\row
			\begin{pmatrix}
				\mathcal L\\
				CA\\
				C
			\end{pmatrix}. \nonumber 
	\end{equation}
	Hence $\mathcal L$ satisfies condition {\rm (A)}.
\end{thm}

\begin{proof}
	Define
	\(
	\mathscr S
	=
	\row
	\begin{pmatrix}
		\mathcal L\\
		CA\\
		C
	\end{pmatrix}.
	\)
	We prove the result in two steps:
	\(
	\row(LA)\subseteq\mathscr S,
	\,\,
	\row(RA)\subseteq\mathscr S.
	\)
	
	\medskip
	\noindent
	\textit{Step 1: Show that $\row(LA)\subseteq\mathscr S$.}
	
	For each $j\in\{1,\ldots,r\}$, by definition of the predecessor matrix
	$\mathcal M_{p+j}^{-}$, every retained row in the power-$\eta_j$
	block
	\[
	\mathcal G_{p+j}^{(\eta_j)}A^{\eta_j}
	\]
	was retained before $L_jA^{\eta_j}$ was examined. At the fixed
	power $\eta_j$, candidates are examined in the order
	\[
	C_1,\ldots,C_p,L_1,\ldots,L_r.
	\]
Therefore, $\mathcal G_{p+j}^{(\eta_j)}$ can contain only output
generators $C_i$ and preceding functional generators
$L_1,\ldots,L_{j-1}$. Hence

$$
\row\!\left(
\mathcal G_{p+j}^{(\eta_j)}
\right)
\subseteq
\row
\begin{pmatrix}
	C\\
	L_1\\
	\vdots\\
	L_{j-1}
\end{pmatrix},
$$

and therefore
\begin{equation}
	\row\!\left(
	\mathcal G_{p+j}^{(\eta_j)}A
	\right)
	\subseteq
	\row
	\begin{pmatrix}
		CA\\
		L_1A\\
		\vdots\\
		L_{j-1}A
	\end{pmatrix}.
	\label{eq:terminal-G-containment}
\end{equation}
When $j=1$, the functional block

$$
\begin{pmatrix}
	L_1A\\
	\vdots\\
	L_{j-1}A
\end{pmatrix}
$$
is absent from the right-hand side of
\eqref{eq:terminal-G-containment}.

	We now prove by induction on $j$ that
	\[
	L_jA\in\mathscr S,
	\qquad j=1,\ldots,r.
	\]
	
	\medskip
	\noindent
	\emph{Base case: $j=1$.}
	
	When $j=1$, there are no preceding functional generators, so
	\eqref{eq:terminal-G-containment} reduces to
	\[
	\row\!\left(
	\mathcal G_{p+1}^{(\eta_1)}A
	\right)
	\subseteq
	\row(CA)
	\subseteq
	\mathscr S.
	\]
	Hence
	\begin{equation}
		\phi_{p+1}^{(\eta_1)}
		\mathcal G_{p+1}^{(\eta_1)}A
		\in\mathscr S. 
		\label{eq:terminal-G-base}
	\end{equation}
	
	If $\eta_1>1$, then $\ell_{1,0}=L_1$, and
	\eqref{eq:functional-chain-recursion} with $t=0$ gives
	\[
	L_1A
	=
	\ell_{1,1}
	+
	\phi_{p+1}^{(\eta_1)}
	\mathcal G_{p+1}^{(\eta_1)}A.
	\]
	The first term belongs to $\mathscr S$ because
	$\ell_{1,1}$ is a row of $\mathcal L$, while the second term
	belongs to $\mathscr S$ by \eqref{eq:terminal-G-base}. Therefore,
	\[
	L_1A\in\mathscr S.
	\]
	
	If $\eta_1=1$, then $L_1=\ell_{1,0}$ is already the terminal row,
	and \eqref{eq:functional-chain-terminal} gives
	\[
	L_1A
	=
	\phi_{p+1}^{(0)}
	\mathcal G_{p+1}^{(0)}
	+
	\phi_{p+1}^{(1)}
	\mathcal G_{p+1}^{(1)}A.
	\]
	Since every row of $\mathcal G_{p+1}^{(0)}$ is an original
	generator $C_i$ or $L_i$,
	\[
	\row\!\left(\mathcal G_{p+1}^{(0)}\right)
	\subseteq
	\row(C,L)
	\subseteq
	\row(C,\mathcal L)
	\subseteq
	\mathscr S,
	\]
	we have
	\[
	\phi_{p+1}^{(0)}
	\mathcal G_{p+1}^{(0)}
	\in
	\row(C,L)
	\subseteq
	\mathscr S.
	\]
	Moreover, since $\eta_1=1$, \eqref{eq:terminal-G-base} gives
	\[
	\phi_{p+1}^{(1)}
	\mathcal G_{p+1}^{(1)}A
	\in\mathscr S.
	\]
	Thus
	\[
	L_1A\in\mathscr S.
	\]
	
	This proves the base case.
	
	\medskip
	\noindent
	\emph{Induction step.}
	
	Let $j\in\{2,\ldots,r\}$ and assume that
	\[
	L_1A,\ldots,L_{j-1}A\in\mathscr S.
	\]
	Then, since $\row(CA)\subseteq\mathscr S$,
	\eqref{eq:terminal-G-containment} implies
	\begin{equation}
		\phi_{p+j}^{(\eta_j)}
		\mathcal G_{p+j}^{(\eta_j)}A
		\in\mathscr S. 
		\label{eq:terminal-G-in-S}
	\end{equation}
	
	If $\eta_j>1$, then $\ell_{j,0}=L_j$, and
	\eqref{eq:functional-chain-recursion} with $t=0$ gives
	\[
	L_jA
	=
	\ell_{j,1}
	+
	\phi_{p+j}^{(\eta_j)}
	\mathcal G_{p+j}^{(\eta_j)}A.
	\]
	The first term belongs to $\mathscr S$ because
	$\ell_{j,1}$ is a row of $\mathcal L$, while the second term
	belongs to $\mathscr S$ by \eqref{eq:terminal-G-in-S}. Hence
	\[
	L_jA\in\mathscr S.
	\]
	
	If $\eta_j=1$, then $L_j=\ell_{j,0}$ is already the terminal row,
	and \eqref{eq:functional-chain-terminal} gives
	\[
	L_jA
	=
	\phi_{p+j}^{(0)}
	\mathcal G_{p+j}^{(0)}
	+
	\phi_{p+j}^{(1)}
	\mathcal G_{p+j}^{(1)}A.
	\]
	Since every row of $\mathcal G_{p+j}^{(0)}$ is an original
	generator $C_i$ or $L_i$,
	\[
	\row\!\left(\mathcal G_{p+j}^{(0)}\right)
	\subseteq
	\row(C,L)
	\subseteq
	\row(C,\mathcal L)
	\subseteq
	\mathscr S,
	\]
	the first term satisfies
	\[
	\phi_{p+j}^{(0)}
	\mathcal G_{p+j}^{(0)}
	\in
	\row(C,L)
	\subseteq
	\mathscr S.
	\]
	The second term
	\[
	\phi_{p+j}^{(1)}
	\mathcal G_{p+j}^{(1)}A
	\]
	belongs to $\mathscr S$ by \eqref{eq:terminal-G-in-S}, since in
	this case $\eta_j=1$. Therefore,
	\[
	L_jA\in\mathscr S.
	\]
	
	Thus the induction step is proved. Hence, by induction,
	\[
	L_jA\in\mathscr S
	\qquad
	\text{for all }j=1,\ldots,r, 
	\]
	and therefore
	\begin{equation}
			\row(LA)\subseteq\mathscr S. 
		\label{eq:LA-contained}
	\end{equation}
	
	\medskip
	\noindent
	\textit{Step 2: Show that $\row(RA)\subseteq\mathscr S$.}
	
	Consider first a nonterminal augmented row $\ell_{j,t}$, where
	$1\leq t\leq\eta_j-2$, whenever this index set is nonempty. By
	\eqref{eq:functional-chain-recursion},
	\[
	\ell_{j,t}A
	=
	\ell_{j,t+1}
	+
	\phi_{p+j}^{(\eta_j-t)}
	\mathcal G_{p+j}^{(\eta_j-t)}A.
	\]
	The first term belongs to $\mathscr S$ because
	$\ell_{j,t+1}$ is a row of $\mathcal L$.
	
	Every row of $\mathcal G_{p+j}^{(\eta_j-t)}$ is an original
	generator $C_i$ or $L_i$. Hence
	\[
	\row\!\left(
	\mathcal G_{p+j}^{(\eta_j-t)}A
	\right)
	\subseteq
	\row(CA,LA).
	\]
	Therefore,
	\[
	\phi_{p+j}^{(\eta_j-t)}
	\mathcal G_{p+j}^{(\eta_j-t)}A
	\in
	\row(CA,LA)
	\subseteq
	\mathscr S,
	\]
	where the last inclusion follows from the definition of
	$\mathscr S$ and \eqref{eq:LA-contained}. Thus
	\[
	\ell_{j,t}A\in\mathscr S.
	\]
	
	It remains to consider the terminal augmented row
	$\ell_{j,\eta_j-1}$ when $\eta_j>1$. By
	\eqref{eq:functional-chain-terminal},
	\[
	\ell_{j,\eta_j-1}A
	=
	\phi_{p+j}^{(0)}
	\mathcal G_{p+j}^{(0)}
	+
	\phi_{p+j}^{(1)}
	\mathcal G_{p+j}^{(1)}A.
	\]
	Since every row of $\mathcal G_{p+j}^{(0)}$ and
	$\mathcal G_{p+j}^{(1)}$ is an original generator $C_i$ or $L_i$,
	\[
	\row\!\left(\mathcal G_{p+j}^{(0)}\right)
	\subseteq \row(C,L),
	\quad
	\row\!\left(\mathcal G_{p+j}^{(1)}A\right)
	\subseteq \row(CA,LA).
	\]
	Consequently,
	\[
	\phi_{p+j}^{(0)}
	\mathcal G_{p+j}^{(0)}
	\in
	\row(C,L)
	\subseteq\mathscr S,
	\]
	and
	\[
	\phi_{p+j}^{(1)}
	\mathcal G_{p+j}^{(1)}A
	\in
	\row(CA,LA)
	\subseteq\mathscr S,
	\]
	where the latter inclusion follows from the definition of
	$\mathscr S$ and \eqref{eq:LA-contained}. Therefore,
	\[
	\ell_{j,\eta_j-1}A\in\mathscr S.
	\]
	
	Thus every augmented row in $R$ has its $A$-image in
	$\mathscr S$, and hence
	\[
	\row(RA)\subseteq\mathscr S.
	\]
	
	Finally, since
	\[
	\mathcal L
	=
	\begin{pmatrix}
		L\\
		R
	\end{pmatrix},
	\]
	we obtain
	\[
	\row(\mathcal LA)
	=
	\row
	\begin{pmatrix}
		LA\\
		RA
	\end{pmatrix}
	\subseteq
	\mathscr S
	=
	\row
	\begin{pmatrix}
		\mathcal L\\
		CA\\
		C
	\end{pmatrix}.
	\]
	Therefore, condition {\rm (A)} holds.
\end{proof}

Theorem~\ref{prop:augmentation-satisfies-A} establishes that the constructed
augmented functional matrix $\mathcal L$ satisfies condition $(A)$.
The next result shows that its augmentation dimension is minimum.
Although this minimum dimension is fixed, the corresponding augmentation
need not be unique: other augmented functional matrices may also satisfy
condition $(A)$ with the same minimum augmentation dimension.

This nonuniqueness becomes important when condition $(S)$ is considered.
A particular minimum augmentation satisfying condition $(A)$ may fail
condition $(S)$, while another minimum augmentation of the same dimension
may satisfy it. It is therefore necessary to characterize the complete
family of minimum augmentations satisfying condition $(A)$ before
testing condition $(S)$.

\subsection{Minimum augmentation for condition (A)} \label{subsec:minimum-augmentation-A}

The augmentation construction of Section~\ref{subsec:augmentation-construction-A} provides an augmented functional matrix satisfying condition {\rm (A)}. We now show that its augmentation dimension is minimum.

\begin{thm}[Minimum augmentation for condition (A)]
	\label{thm:minimum-A-augmentation}
	
	Let $\eta_j$, $j=1,\ldots,r$, denote the functional-observability
	indices associated with the prescribed functional rows $L_j$.
	Then the minimum number of additional functional rows required to
	satisfy condition {\rm (A)} is
	\begin{equation}
			(q_0-r)
			=
			\sum_{j=1}^{r}(\eta_j-1). \nonumber 
	\end{equation}
	Equivalently, the minimum dimension of an augmented functional matrix
	satisfying condition {\rm (A)} is
	\begin{equation}
			q_0
			=
			\sum_{j=1}^{r}\eta_j. \nonumber 
	\end{equation}
	
\end{thm}

\begin{proof}
	We prove the result by establishing matching upper and lower bounds.
	
	\medskip
	\noindent
	\textit{Upper bound.}
	For each prescribed functional row $L_j$, the construction in
	Subsection~\ref{subsec:augmentation-construction-A} introduces
	$\eta_j-1$ additional rows. Hence the total number of augmented
	rows is
	$$
	\sum_{j=1}^{r}(\eta_j-1).
	$$
	By Theorem~\ref{prop:augmentation-satisfies-A}, the resulting
	augmented functional matrix satisfies condition {\rm (A)}.
	Therefore,
	\begin{equation}
		(q_0-r)
		\leq
		\sum_{j=1}^{r}(\eta_j-1). 
		\label{eq:A-upper-bound}
	\end{equation}
	
	\medskip
	\noindent
	
\textit{Lower bound.}
Let
$$
\widetilde{\mathcal L}
=
\begin{pmatrix}
	L\\
	\widetilde R
\end{pmatrix},
\qquad
\widetilde R\in\mathbb R^{\widetilde q\times n},
$$
be any augmentation satisfying condition {\rm (A)}. We show that
$$
\widetilde q
\geq
\sum_{j=1}^{r}(\eta_j-1).
$$

Condition {\rm (A)} implies that, for suitable matrices $T$, $G_0$,
and $G_1$,
$$
\widetilde{\mathcal L}A
=
T\widetilde{\mathcal L}
+
G_0C
+
G_1CA.
$$
Define
$$
\overline{\mathcal L}
=
\widetilde{\mathcal L}-G_1C.
$$
Then
$$
\overline{\mathcal L}A
=
T\overline{\mathcal L}+GC,
\qquad
G:=TG_1+G_0,
$$
and
$$
\row
\begin{pmatrix}
	\overline{\mathcal L}\\
	C
\end{pmatrix}
=
\row
\begin{pmatrix}
	\widetilde{\mathcal L}\\
	C
\end{pmatrix}.
$$
Since the rows of $L$ are contained in $\widetilde{\mathcal L}$,
there exist matrices $P$ and $E$ such that
$$
L=P\overline{\mathcal L}+EC.
$$

Repeatedly using
$\overline{\mathcal L}A=T\overline{\mathcal L}+GC$ therefore gives,
for every $k\geq0$,
\begin{equation}
	L_jA^k
	=
	w_{j,k}\overline{\mathcal L}
	+
	h_{j,k},
	\label{eq:LjAk-wjk}
\end{equation}
where
$$
\row(h_{j,k})
\subseteq
\row
\begin{pmatrix}
	C\\
	CA\\
	\vdots\\
	CA^k
\end{pmatrix}.
$$

Now consider the functional candidate rows retained by the global
rank-selection procedure:
$$
L_jA^k,
\qquad
j=1,\ldots,r,
\qquad
0\leq k<\eta_j.
$$
We claim that the corresponding coefficient rows $w_{j,k}$ are
linearly independent.

Suppose otherwise, and let $w_{j,k}$ be the first such coefficient
row, in the global examination order, that is a linear combination
of preceding coefficient rows. Then there exist scalars
$\gamma_{i,q}$ such that
$$
w_{j,k}
=
\sum_{(i,q)\prec(j,k)}
\gamma_{i,q}w_{i,q},
$$
where $(i,q)\prec(j,k)$ means that the functional candidate $L_iA^q$
is examined before $L_jA^k$ in the global order. Substituting this
relation into \eqref{eq:LjAk-wjk} and using
$$
w_{i,q}\overline{\mathcal L}
=
L_iA^q-h_{i,q}
$$
for the preceding functional candidates gives
$$
\begin{aligned}
	L_jA^k
	&=
	\sum_{(i,q)\prec(j,k)}
	\gamma_{i,q}w_{i,q}\overline{\mathcal L}
	+h_{j,k}\\
	&=
	\sum_{(i,q)\prec(j,k)}
	\gamma_{i,q}L_iA^q
	+
	h_{j,k}
	-
	\sum_{(i,q)\prec(j,k)}
	\gamma_{i,q}h_{i,q}.
\end{aligned}
$$

Every functional candidate $L_iA^q$ appearing on the right-hand side
has been examined before $L_jA^k$. Moreover, since $q\leq k$, all
output-generated rows occurring in $h_{i,q}$ and $h_{j,k}$ have also
been examined before $L_jA^k$. Every previously examined row belongs
to the span of the rows retained up to that point, whether it was
itself retained or rejected. Hence $L_jA^k$ belongs to the span of
the rows retained before it is examined.

Let $\mathcal M^-$ denote the matrix of rows retained immediately
before $L_jA^k$ is examined. Then
$$
\rank
\begin{pmatrix}
	\mathcal M^-\\
	L_jA^k
\end{pmatrix}
=
\rank(\mathcal M^-).
$$
By the rank-retention rule, $L_jA^k$ would therefore not be retained.
This contradicts $k<\eta_j$, since, by definition of $\eta_j$,
$$
L_j,\;
L_jA,\;
\ldots,\;
L_jA^{\eta_j-1}
$$
are all retained. Therefore the coefficient rows
$$
w_{j,k},
\qquad
j=1,\ldots,r,
\qquad
0\leq k<\eta_j,
$$
are linearly independent.
There are
$$
\sum_{j=1}^{r}\eta_j
$$
such rows. Since $\overline{\mathcal L}$ has $r+\widetilde q$ rows,
each $w_{j,k}$ belongs to
$\mathbb R^{1\times(r+\widetilde q)}$. Hence
$$
\sum_{j=1}^{r}\eta_j
\leq
r+\widetilde q,
$$
and therefore
\begin{equation}
	\widetilde q
	\geq
	\sum_{j=1}^{r}(\eta_j-1).
	\label{eq:A-lower-bound}
\end{equation}
Since $\widetilde{\mathcal L}$ was arbitrary, this lower bound holds
for every augmentation satisfying condition {\rm (A)}. Combining
\eqref{eq:A-upper-bound} and \eqref{eq:A-lower-bound} gives
$$
(q_0-r)
=
\sum_{j=1}^{r}(\eta_j-1).
$$
Consequently,
$$
q_0
=
\sum_{j=1}^{r}\eta_j.
$$
\end{proof}

\begin{remark}[Nonuniqueness of the minimum augmentation]
	\label{rem:nonunique-minimum-A}
	
	Theorem~\ref{thm:minimum-A-augmentation} determines uniquely the
	minimum augmentation dimension
	\begin{equation}
		(q_0-r)
		=
		\sum_{j=1}^{r}(\eta_j-1), \nonumber 
	\end{equation}
	but it does not imply uniqueness of the corresponding minimum
	augmentation subspace.

	In general, distinct augmented functional spaces of the same minimum
	dimension may satisfy condition {\rm (A)}. The ordered
	rank-retention construction selects a particular minimum
	augmentation determined by the prescribed generator and examination
	orders.
	Thus the minimum dimension is unique, whereas the minimizing
	augmentation need not be.
\end{remark}

\section{Complete Family of Minimum Condition-$(A)$ Augmentations}
\label{sec:complete-minimum-A-augmentations}

Section~\ref{sec:augmentation-construction} constructs one minimum
augmentation satisfying condition $(A)$ and establishes that the
minimum functional dimension is
$$
q_0
=
\sum_{j=1}^{r}\eta_j.
$$
The minimum augmentation, however, need not be unique. We now
characterize all minimum augmentations satisfying condition $(A)$.

Let
$$
\mathcal L
=
\begin{pmatrix}
	L\\
	R
\end{pmatrix}
$$
be an augmentation of the prescribed functional matrix $L$. Since the
measured quantity $Cx$ is already available, components belonging to \(\operatorname{row}(C)\) do not contribute additional functional directions.

Let
$$
d_C:=n-p,
$$
and choose
$$
N_C\in\mathbb R^{n\times d_C}
$$
whose columns form a basis of $\ker C$. Then
$$
CN_C=0,
\qquad
\rank(N_C)=d_C.
$$
Choose also a left inverse
$$
M_C\in\mathbb R^{d_C\times n}
$$
such that
$$
M_CN_C=I_{d_C}.
$$
Multiplication by $N_C$ factors out precisely the measured directions,
since
$$
vN_C=0
\quad\Longleftrightarrow\quad
v\in\operatorname{row}(C).
$$
Consequently,
$$
\rank
\begin{pmatrix}
	C\\
	\mathcal L
\end{pmatrix}
-p
=
\rank(\mathcal LN_C).
$$
Thus $\rank(\mathcal LN_C)$ is the number of independent functional
directions represented by $\mathcal L$ modulo the measured row space
$\operatorname{row}(C)$.

By the standing assumption
$$
\rank
\begin{pmatrix}
	C\\
	L
\end{pmatrix}
=
p+r,
$$
we have
$$
\rank(LN_C)=r.
$$
Since a minimum augmentation has functional dimension
$$
q_0
=
\sum_{j=1}^{r}\eta_j,
$$
it requires
$$
k
:=
(q_0-r)
=
\sum_{j=1}^{r}(\eta_j-1)
$$
additional independent directions beyond those represented by $L$.
Accordingly, we seek minimum augmentations satisfying
$$
\rank(\mathcal LN_C)
=
r+k
=
q_0.
$$
Since such an $\mathcal L$ has exactly $r+k$ rows, this also implies
$$
\rank(\mathcal L)=r+k=q_0.
$$
We first describe all possible choices of these $k$ additional
directions and then impose condition $(A)$.

\subsection{Parameterization of the additional directions}
\label{subsec:parameterization-additional-directions}

Since
$$
\rank(LN_C)=r,
$$
there exists a nonsingular matrix
$$
T\in\mathbb R^{d_C\times d_C}
$$
such that
$$
LN_CT
=
\begin{pmatrix}
	I_r&0
\end{pmatrix}.
$$
Thus, in the coordinates defined by $T$, the prescribed functional
directions occupy the first $r$ coordinates.

Set
$$
s:=d_C-r=n-p-r.
$$
A minimum augmentation must add $k$ independent directions to the
$r$ prescribed directions. By subtracting suitable linear combinations
of the prescribed rows from the additional rows, every minimum candidate
row space can be represented in the transformed coordinates as
$$
\begin{pmatrix}
	I_r&0\\
	0&Y
\end{pmatrix},
\qquad
Y\in\mathbb R^{k\times s},
\qquad
\rank(Y)=k.
$$
Let
$$
J\subseteq\{1,\ldots,s\},
\quad
|J|=k,
$$
denote a set of $k$ pivot-column indices.
For $k=0$, take
$J=\varnothing$ and no parameter $\Theta$ is required. For each
$k>0$ and such $J$, choose a permutation matrix $\Pi_J$
such that, in
$$
Y_J(\Theta)
=
\begin{pmatrix}
	I_k&\Theta
\end{pmatrix}
\Pi_J^T,
$$
the columns originating from $I_k$ occupy the positions indexed by
$J$, while
$$
\Theta\in\mathbb R^{k\times(s-k)}
$$
contains the remaining free columns. Since the columns indexed by $J$
are the columns of $I_k$, they are linearly independent, and therefore
$$
\rank Y_J(\Theta)=k
$$
for every $\Theta$.

Conversely, let
$$
Y\in\mathbb R^{k\times s},
\quad
\rank(Y)=k.
$$
Then $Y$ contains $k$ linearly independent columns. Let $J$ denote
their indices. A nonsingular row transformation can normalize these
columns to $I_k$, while the remaining $s-k$ columns form a matrix
$\Theta\in\mathbb R^{k\times(s-k)}$. Hence
$$
\operatorname{row}(Y)
=
\operatorname{row}\bigl(Y_J(\Theta)\bigr)
$$
for some $J$ and $\Theta$. Consequently, as $J$ ranges over all
$k$-element subsets of $\{1,\ldots,s\}$ and $\Theta$ ranges over
$\mathbb R^{k\times(s-k)}$, the matrices $Y_J(\Theta)$ represent all
$k$-dimensional row spaces in $\mathbb R^s$.

For each pair $(J,\Theta)$, define
$$
W_J(\Theta)
:=
\begin{pmatrix}
	I_r&0\\
	0&Y_J(\Theta)
\end{pmatrix}.
$$
Since
$$
\rank Y_J(\Theta)=k,
$$
it follows that
$$
\rank W_J(\Theta)
=
r+k
=
q_0.
$$
Thus every $W_J(\Theta)$ has the required minimum functional dimension.
It remains only to determine which members of this family satisfy
condition $(A)$.

\subsection{Imposing condition $(A)$}
\label{subsec:imposing-condition-A}
The preceding parameterization describes all minimum-dimensional
candidates containing the prescribed functional directions. We now
restrict this family to those candidates that satisfy condition $(A)$.
To do so, we first express condition $(A)$ entirely in the reduced
coordinates introduced above.
\begin{lemma}[Reduced form of condition $(A)$]
	\label{lem:reduced-condition-A}
	Let
	$$
	\bar A=M_CAN_C,
	\qquad
	B_A=CAN_C.
	$$
	Then condition $(A)$,
	$$
	\operatorname{row}(\mathcal LA)
	\subseteq
	\operatorname{row}
	\begin{pmatrix}
		\mathcal L\\
		CA\\
		C
	\end{pmatrix},
	$$
	is equivalent to
	$$
	\operatorname{row}\bigl((\mathcal LN_C)\bar A\bigr)
	\subseteq
	\operatorname{row}
	\begin{pmatrix}
		\mathcal LN_C\\
		B_A
	\end{pmatrix}.
	$$
\end{lemma}

\begin{proof}
	To verify the claim, note first that for any row vector $v$,
	$$
	v-(vN_C)M_C
	$$
	annihilates $N_C$, since
	$$
	\bigl(v-(vN_C)M_C\bigr)N_C
	=
	vN_C-vN_CM_CN_C
	=
	0.
	$$
	Since the columns of \(N_C\) span \(\ker C\), a row vector \(v\) satisfies \(vN_C=0\) if and only if \(v\in\row(C)\), and 
	it follows that
	$$
	v-(vN_C)M_C\in\operatorname{row}(C).
	$$
	Hence there exists a row vector $\gamma\in\mathbb R^{1\times p}$ such that
	$$
	v-(vN_C)M_C=\gamma C,
	$$
	or equivalently,
	$$
	v=(vN_C)M_C+\gamma C.
	$$
	Applying this rowwise to $\mathcal L$ gives
	$$
	\mathcal L
	=
	(\mathcal LN_C)M_C+\Gamma C
	$$
	for some matrix $\Gamma$. Hence
	$$
	\mathcal LA
	=
	(\mathcal LN_C)M_CA+\Gamma CA.
	$$
	Multiplying by $N_C$ on the right yields
	$$
	\mathcal LAN_C
	=
	(\mathcal LN_C)M_CAN_C+\Gamma CAN_C,
	$$
	and therefore
	$$
	\mathcal LAN_C
	=
	(\mathcal LN_C)\bar A+\Gamma B_A.
	$$
	
	Now suppose condition $(A)$ holds, that is,
	$$
	\operatorname{row}(\mathcal LA)
	\subseteq
	\operatorname{row}
	\begin{pmatrix}
		\mathcal L\\
		CA\\
		C
	\end{pmatrix}.
	$$
	Multiplying by $N_C$ gives
	$$
	\operatorname{row}(\mathcal LAN_C)
	\subseteq
	\operatorname{row}
	\begin{pmatrix}
		\mathcal LN_C\\
		B_A
	\end{pmatrix},
	$$
	since $CN_C=0$. Using
	$$
	\mathcal LAN_C
	=
	(\mathcal LN_C)\bar A+\Gamma B_A,
	$$
	we obtain
	$$
	\operatorname{row}((\mathcal LN_C)\bar A)
	\subseteq
	\operatorname{row}
	\begin{pmatrix}
		\mathcal LN_C\\
		B_A
	\end{pmatrix}.
	$$
	
	Conversely, suppose
	$$
	\operatorname{row}((\mathcal LN_C)\bar A)
	\subseteq
	\operatorname{row}
	\begin{pmatrix}
		\mathcal LN_C\\
		B_A
	\end{pmatrix}.
	$$
	Then, by the identity above,
	$$
	\operatorname{row}(\mathcal LAN_C)
	\subseteq
	\operatorname{row}
	\begin{pmatrix}
		\mathcal LN_C\\
		B_A
	\end{pmatrix}.
	$$
	Hence there exist matrices $F$ and $G$ such that
	$$
	\mathcal LAN_C
	=
	F\mathcal LN_C+GB_A.
	$$
	Equivalently,
	$$
	(\mathcal LA-F\mathcal L-GCA)N_C=0.
	$$
	Therefore
	$$
	\operatorname{row}(\mathcal LA-F\mathcal L-GCA)
	\subseteq
	\operatorname{row}(C),
	$$
	and thus
	$$
	\operatorname{row}(\mathcal LA)
	\subseteq
	\operatorname{row}
	\begin{pmatrix}
		\mathcal L\\
		CA\\
		C
	\end{pmatrix}.
	$$
	This is precisely condition $(A)$.
\end{proof}

As established above, condition $(A)$ in the reduced coordinates is
equivalent to
$$
\operatorname{row}(\mathcal LN_C\bar A)
\subseteq
\operatorname{row}
\begin{pmatrix}
	\mathcal LN_C\\
	B_A
\end{pmatrix}.
$$
The matrices $W_J(\Theta)$ are expressed in the transformed coordinates
defined by $T$. Therefore the same coordinate transformation must be
applied to $\bar A$ and $B_A$. Define
$$
\bar A_T
:=
T^{-1}\bar A T,
\qquad
B_{A,T}
:=
B_AT.
$$
Condition $(A)$ for the candidate $W_J(\Theta)$ is then
$$
\operatorname{row}\bigl(W_J(\Theta)\bar A_T\bigr)
\subseteq
\operatorname{row}
\begin{pmatrix}
	W_J(\Theta)\\
	B_{A,T}
\end{pmatrix}.
$$
This inclusion holds if and only if adding the rows of
$W_J(\Theta)\bar A_T$ does not increase the rank. Hence condition $(A)$
is equivalent to
\begin{equation}
		\rank
		\begin{pmatrix}
			W_J(\Theta)\\
			B_{A,T}\\
			W_J(\Theta)\bar A_T
		\end{pmatrix}
		=
		\rank
		\begin{pmatrix}
			W_J(\Theta)\\
			B_{A,T}
		\end{pmatrix}. 	\label{eq:family-rank-test}
\end{equation}
For each pivot set $J$, define
\begin{IEEEeqnarray}{rcl}
\Omega_{A,J}^\star
:=&&
\Bigg\{
\Theta\in\mathbb R^{k\times(s-k)}:  \nonumber\\
&& \rank
\begin{pmatrix}
	W_J(\Theta)\\
	B_{A,T}\\
	W_J(\Theta)\bar A_T
\end{pmatrix}
=
\rank
\begin{pmatrix}
	W_J(\Theta)\\
	B_{A,T}
\end{pmatrix}
\Bigg\}. \nonumber 
\end{IEEEeqnarray}
Thus $\Omega_{A,J}^\star$ contains exactly those values of $\Theta$
for which the corresponding minimum candidate satisfies condition
$(A)$.

\subsection{Completeness of the parameterization}
\label{subsec:completeness-minimum-A}

We now show that the preceding construction produces every minimum
candidate satisfying condition $(A)$.


\begin{thm}
	\label{prop:complete-minimum-A-family}
	For every
	\[
	J\subseteq\{1,\ldots,s\},
	\qquad |J|=k,
	\]
	and every
	\[
	\Theta\in\Omega_{A,J}^\star,
	\]
	the matrix \(W_J(\Theta)\) is a reduced augmentation of minimum
	functional dimension satisfying condition~\((A)\).
	
	Conversely, every reduced augmentation of minimum functional
	dimension satisfying condition~\((A)\) is row-equivalent to
	\(W_J(\Theta)\) for at least one pair
	\[
	(J,\Theta),
	\qquad
	|J|=k,
	\qquad
	\Theta\in\Omega_{A,J}^\star.
	\]
	Consequently,
	\[
	\left\{
	W_J(\Theta):
	J\subseteq\{1,\ldots,s\},\
	|J|=k,\
	\Theta\in\Omega_{A,J}^\star
	\right\}
	\]
	represents, up to row equivalence, the complete family of reduced
	augmentations of minimum functional dimension satisfying
	condition~\((A)\).
\end{thm}

\begin{proof}
	First, fix any
	\[
	J\subseteq\{1,\ldots,s\},
	\qquad |J|=k,
	\]
	and any
	\[
	\Theta\in\Omega_{A,J}^\star.
	\]
	By construction,
	\[
	\rank\bigl(Y_J(\Theta)\bigr)=k.
	\]
	Therefore,
	\[
	\rank\bigl(W_J(\Theta)\bigr)
	=
	r+k
	=
	q_0.
	\]
	Hence \(W_J(\Theta)\) is a reduced augmentation of minimum
	functional dimension.
	Moreover, since
	\[
	\Theta\in\Omega_{A,J}^\star,
	\]
	the defining rank condition for \(\Omega_{A,J}^\star\) gives
	\[
	\rank
	\begin{pmatrix}
		W_J(\Theta)\\
		B_{A,T}\\
		W_J(\Theta)\bar A_T
	\end{pmatrix}
	=
	\rank
	\begin{pmatrix}
		W_J(\Theta)\\
		B_{A,T}
	\end{pmatrix}.
	\]
	Thus \(W_J(\Theta)\) satisfies condition~\((A)\). Consequently,
	every matrix in the stated family is a reduced augmentation of
	minimum functional dimension satisfying condition~\((A)\).
	
	Conversely, consider any reduced augmentation of minimum functional
	dimension satisfying condition~\((A)\). After multiplication by
	\(N_C\) and application of the coordinate transformation \(T\),
	its reduced rows have rank
	\[
	q_0=r+k
	\]
	and contain the prescribed rows
$
	\begin{pmatrix}
		I_r&0
	\end{pmatrix}.
$
	By elementary row operations, its reduced row space can therefore
	be represented in the form
	\[
	\begin{pmatrix}
		I_r&0\\
		0&Y
	\end{pmatrix},
	\]
	where
$
	Y\in\mathbb R^{k\times s},
	\quad
	\rank(Y)=k.
$
	
	Since \(Y\) has rank \(k\), it contains at least one set of \(k\)
	linearly independent columns. Let
	\[
	J\subseteq\{1,\ldots,s\},
	\quad |J|=k,
	\]
	denote the indices of any such set. Using the columns indexed by
	\(J\) as pivot columns, the row space of \(Y\) admits a chart
	representation of the form
	\[
	\row(Y)
	=
	\row\bigl(Y_J(\Theta)\bigr)
	\]
	for some
$
	\Theta\in\mathbb R^{k\times(s-k)}.
$
	It follows that the original reduced augmentation is row-equivalent
	to \(W_J(\Theta)\).
	
	Since the original augmentation satisfies condition~\((A)\), and
	condition~\((A)\) depends only on the corresponding reduced row
	space, its representative \(W_J(\Theta)\) also satisfies the
	transformed rank condition
	\[
	\rank
	\begin{pmatrix}
		W_J(\Theta)\\
		B_{A,T}\\
		W_J(\Theta)\bar A_T
	\end{pmatrix}
	=
	\rank
	\begin{pmatrix}
		W_J(\Theta)\\
		B_{A,T}
	\end{pmatrix}.
	\]
	Therefore, by the definition of \(\Omega_{A,J}^\star\),
$
	\Theta\in\Omega_{A,J}^\star.
$
	
	Thus every reduced augmentation of minimum functional dimension
	satisfying condition~\((A)\) is represented, up to row equivalence,
	by at least one matrix \(W_J(\Theta)\) with
	\[
	J\subseteq\{1,\ldots,s\},
	\qquad
	|J|=k,
	\qquad
	\Theta\in\Omega_{A,J}^\star.
	\]
	Together with the first part of the proof, this shows that
	\[
	\left\{
	W_J(\Theta):
	J\subseteq\{1,\ldots,s\},\
	|J|=k,\
	\Theta\in\Omega_{A,J}^\star
	\right\}
	\]
	represents, up to row equivalence, the complete family of reduced
	augmentations of minimum functional dimension satisfying
	condition~\((A)\).
\end{proof}

\subsection{Recovery of the augmentation matrix}
\label{subsec:recovery-minimum-A}

The matrices $W_J(\Theta)$ are expressed in transformed reduced
coordinates. We now recover the corresponding augmentation

$$
\mathcal L
=
\begin{pmatrix}
	L\\
	R
\end{pmatrix}
$$
in the original state coordinates.

For each admissible pair $(J,\Theta)$, returning to the original
reduced coordinates gives

$$
W_J(\Theta)T^{-1}
=
\begin{pmatrix}
	LN_C\\
	X_J(\Theta)
\end{pmatrix},
$$
which defines
$$
X_J(\Theta)\in\mathbb R^{k\times d_C}.
$$
Thus the first $r$ rows are $LN_C$, while $X_J(\Theta)$ contains the
$k$ additional reduced directions. 
Recovering the augmentation therefore amounts to solving
\begin{equation}
	RN_C=X_J(\Theta).
	\label{eq:recovery-RNC}
\end{equation}
Since $N_C$ has full column rank, this equation is solvable for every
$X_J(\Theta)$. Indeed, for any left inverse $M_C$ satisfying
$M_CN_C=I_{d_C}$, a particular solution is
$$
R_J(\Theta)=X_J(\Theta)M_C.
$$
since
$$
R_J(\Theta)N_C
=
X_J(\Theta)M_CN_C
=
X_J(\Theta).
$$
Hence
$$
\mathcal L_J(\Theta)
=
\begin{pmatrix}
	L\\
	X_J(\Theta)M_C
\end{pmatrix}
$$
has precisely the reduced rows represented by $W_J(\Theta)$.

The solution of \eqref{eq:recovery-RNC} is not unique. Since
$CN_C=0$, for any
$
K\in\mathbb R^{k\times p},
$
the matrix
$$
R
=
X_J(\Theta)M_C+KC
$$
satisfies
$$
RN_C
=
X_J(\Theta)M_CN_C+KCN_C
=
X_J(\Theta).
$$
Thus adding rows from the measured row space does not change the
reduced augmentation.

Conversely, let $R$ be any solution of \eqref{eq:recovery-RNC}. Then
$$
\bigl(R-X_J(\Theta)M_C\bigr)N_C=0.
$$
Since the columns of $N_C$ form a basis for $\ker C$ and $C$ has full
row rank,
$$
\{v\in\mathbb R^{1\times n}:vN_C=0\}
=
\row(C).
$$
Therefore
$$
R-X_J(\Theta)M_C=KC
$$
for some $K\in\mathbb R^{k\times p}$. Consequently, all matrix
representatives associated with an admissible pair $(J,\Theta)$ are
exactly
$$
R
=
X_J(\Theta)M_C+KC,
\qquad
K\in\mathbb R^{k\times p}.
$$
Hence the complete family of minimum augmented functional matrices
satisfying condition $(A)$ is represented by
\begin{equation}
	\mathcal L
	=
	\begin{pmatrix}
		L\\
		X_J(\Theta)M_C+KC
	\end{pmatrix},
	\label{eq:mathcalL}
\end{equation}
where
$$
J\subseteq\{1,\ldots,s\},
\qquad
|J|=k,
\qquad
\Theta\in\Omega_{A,J}^\star,
\qquad
K\in\mathbb R^{k\times p}.
$$

\begin{remark}
	The minimum condition-$(A)$ family is nonempty. Indeed, the chain construction of Section~\ref{sec:augmentation-construction} explicitly produces a
	minimum augmentation satisfying condition $(A)$. Since the
	parameterization above is complete, there exist at least one pivot
	set $J^\star$ and one parameter value
	$$
	\Theta^\star\in\Omega_{A,J^\star}^\star.
	$$
\end{remark}

\section{Spectral Feasibility at the Minimum Condition-$(A)$ Order}
\label{sec:spectral}

The preceding development characterizes the complete family of minimum
augmentations satisfying condition $(A)$. Therefore, the search for a
minimum functional observer can be restricted to matrices of the form
\eqref{eq:mathcalL}. It remains to determine which members of this
family also satisfy condition $(S)$. We now derive a reduced spectral
test for this purpose.

Set
\begin{equation}
	d_{CA}
	:=
	n-
	\rank
	\begin{pmatrix}
		C\\
		CA
	\end{pmatrix}, \nonumber 
\end{equation}
and choose
\begin{equation}
	N_{CA}\in\mathbb R^{n\times d_{CA}},
	\quad
	\operatorname{im}N_{CA}
	=
	\ker
	\begin{pmatrix}
		C\\
		CA
	\end{pmatrix}. \nonumber
\end{equation}
Thus
$$
CN_{CA}=0,
\quad
CAN_{CA}=0.
$$
For any minimum condition-$(A)$ augmentation $\mathcal L$, define
\begin{equation}
	H:=\mathcal LN_{CA},
	\quad
	F:=\mathcal LAN_{CA},
	\quad
	P_{\mathcal L}(\lambda):=\lambda H-F. \nonumber
\end{equation}

The following rank identity removes the directions already contained in
$\operatorname{row}(C,CA)$.

\begin{lemma}[Rank reduction]
	\label{lem:CA-rank-reduction}
	For any matrix $M$ with $n$ columns,
	\begin{equation}
		\rank
		\begin{pmatrix}
			M\\
			C\\
			CA
		\end{pmatrix}
		=
		\rank
		\begin{pmatrix}
			C\\
			CA
		\end{pmatrix}
		+
		\rank(MN_{CA}). 
		\label{eq:rankreduction}
	\end{equation}
\end{lemma}

\begin{proof}
	Since the columns of $N_{CA}$ form a basis of
	$$
	\ker
	\begin{pmatrix}
		C\\
		CA
	\end{pmatrix},
	$$
	a row vector $v$ satisfies
	$$
	vN_{CA}=0
	$$
	if and only if
	$$
	v\in
	\operatorname{row}
	\begin{pmatrix}
		C\\
		CA
	\end{pmatrix}.
	$$
	Hence multiplication by $N_{CA}$ removes precisely the row directions
	already contained in
	$$
	\operatorname{row}
	\begin{pmatrix}
		C\\
		CA
	\end{pmatrix}.
	$$
	Therefore the number of independent directions contributed by $M$
	beyond this row space is $\rank(MN_{CA})$, which gives
	\eqref{eq:rankreduction}.
\end{proof}

\begin{thm}[Reduced form of condition $(S)$]
	\label{thm:reduced-condition-S}
	Let $\mathcal L$ be a minimum augmentation satisfying condition $(A)$.
	Then condition $(S)$ is equivalent to
	\begin{equation}
			\rank P_{\mathcal L}(\lambda)
			=
			\rank H,
			\qquad
			\forall\lambda\in\mathbb C.
		\label{eq:reducedS}
	\end{equation}
\end{thm}

\begin{proof}
	Applying Lemma~\ref{lem:CA-rank-reduction} with
	$$
	M=\lambda\mathcal L-\mathcal LA
	$$
	gives
	$$
	\rank
	\begin{pmatrix}
		\lambda\mathcal L-\mathcal LA\\
		CA\\
		C
	\end{pmatrix}
	=
	\rank
	\begin{pmatrix}
		C\\
		CA
	\end{pmatrix}
	+
	\rank\bigl(
	(\lambda\mathcal L-\mathcal LA)N_{CA}
	\bigr).
	$$
	By definition,
	$$
	(\lambda\mathcal L-\mathcal LA)N_{CA}
	=
	\lambda H-F
	=
	P_{\mathcal L}(\lambda).
	$$
	Hence
	$$
	\rank
	\begin{pmatrix}
		\lambda\mathcal L-\mathcal LA\\
		CA\\
		C
	\end{pmatrix}
	=
	\rank
	\begin{pmatrix}
		C\\
		CA
	\end{pmatrix}
	+
	\rank P_{\mathcal L}(\lambda).
	$$
	
	Applying the same lemma with $M=\mathcal L$ gives
	$$
	\rank
	\begin{pmatrix}
		\mathcal L\\
		CA\\
		C
	\end{pmatrix}
	=
	\rank
	\begin{pmatrix}
		C\\
		CA
	\end{pmatrix}
	+
	\rank H.
	$$
	Therefore condition $(S)$ is equivalent to
	$$
	\rank P_{\mathcal L}(\lambda)
	=
	\rank H,
	\qquad
	\forall\lambda\in\mathbb C.
	$$
\end{proof}

Condition $(A)$ provides an additional useful structure. Since
$$
\operatorname{row}(\mathcal LA)
\subseteq
\operatorname{row}
\begin{pmatrix}
	\mathcal L\\
	CA\\
	C
\end{pmatrix},
$$
there exist matrices $T_{\mathcal L}$, $G_0$, and $G_1$ such that
$$
\mathcal LA
=
T_{\mathcal L}\mathcal L+G_0C+G_1CA.
$$
Multiplying by $N_{CA}$ and using
$$
CN_{CA}=0,
\quad
CAN_{CA}=0,
$$
gives
$$
F=T_{\mathcal L}H.
$$
Consequently,
\begin{equation}
	P_{\mathcal L}(\lambda)
	=
	(\lambda I-T_{\mathcal L})H. \nonumber
\end{equation}
If
$$
h:=\rank(H),
$$
then $\rank P_{\mathcal L}(\lambda)\leq h$ for every $\lambda$.
Moreover, $\lambda I-T_{\mathcal L}$ is nonsingular for all but
finitely many $\lambda$, for which
$$
\rank P_{\mathcal L}(\lambda)=h.
$$
Thus condition $(S)$ requires
$$
\rank P_{\mathcal L}(\lambda)=h,
\qquad
\forall\lambda\in\mathbb C.
$$

\subsection{Finite algebraic spectral test}
\label{subsec:finite-spectral-test}

The reduced spectral condition \eqref{eq:reducedS} can be tested
algebraically for each member of the minimum condition-$(A)$ family.

Let
$$
h:=\rank H.
$$
If $h=0$, then $H=0$. Since condition $(A)$ gives
$$
F=T_{\mathcal L}H,
$$
it follows that $F=0$. Hence
$$
P_{\mathcal L}(\lambda)=0,
\qquad
\forall\lambda\in\mathbb C,
$$
and therefore
$$
\rank P_{\mathcal L}(\lambda)
=
0
=
\rank H,
\qquad
\forall\lambda\in\mathbb C.
$$
Thus condition $(S)$ is automatically satisfied.

Suppose now that $h>0$, and let
$$
\Delta_1(\lambda),\ldots,\Delta_{N_h}(\lambda)
$$
denote all $h\times h$ minors of $P_{\mathcal L}(\lambda)$, where
$N_h$ denotes the number of such minors. Since
$$
P_{\mathcal L}(\lambda)
=
(\lambda I-T_{\mathcal L})H,
$$
we have
$$
\rank P_{\mathcal L}(\lambda)
=
h
$$
whenever $\lambda I-T_{\mathcal L}$ is nonsingular. Condition $(S)$
therefore holds if and only if
$
\rank P_{\mathcal L}(\lambda)
=
h,
\qquad
\forall\lambda\in\mathbb C.
$
Equivalently, the maximal minors
$$
\Delta_1(\lambda),\ldots,\Delta_{N_h}(\lambda)
$$

have no common zero in $\mathbb C$. Since these are univariate
polynomials, this is equivalent to requiring that they have no
nonconstant common polynomial factor, or equivalently,
\begin{equation}
	\gcd_{\lambda}
	\bigl(
	\Delta_1(\lambda),\ldots,\Delta_{N_h}(\lambda)
	\bigr)
	=
	1,
	\label{eq:gcdtest}
\end{equation}
where $\gcd_{\lambda}$ denotes the greatest common divisor with respect
to the variable $\lambda$, normalized up to multiplication by a
nonzero constant.

Combining this algebraic test with the complete parameterization of the
minimum condition-$(A)$ family gives the following existence criterion.

\begin{thm}[Existence at the minimum condition-$(A)$ order]
	\label{thm:minimum-A-order-existence}
	Let
	$$
	q_0
	=
	\sum_{j=1}^r\eta_j,
	\quad
	k
	=
	(q_0-r)
	=
	\sum_{j=1}^r(\eta_j-1),
	\quad
	s
	=
	n-p-r.
	$$
	A functional observer of order $q_0$ exists with arbitrary observer-pole
	assignment if and only if
	$$
	J\subseteq\{1,\ldots,s\},
	\quad
	|J|=k,
	\quad
	\Theta\in\Omega_{A,J}^\star,
	$$
	and
	$$
	K\in\mathbb R^{k\times p}
	$$
	such that
	$$
	\mathcal L
	=
	\begin{pmatrix}
		L\\
		X_J(\Theta)M_C+KC
	\end{pmatrix}
	$$
	satisfies
	\begin{equation}
		\rank P_{\mathcal L}(\lambda)
		=
		\rank H,
		\qquad
		\forall\lambda\in\mathbb C. \nonumber
	\end{equation}
	If $h:=\rank H=0$, this condition is satisfied automatically. If
	$h>0$, it is equivalent to
	$$
	\gcd_{\lambda}
	\bigl(
	\Delta_1(\lambda),\ldots,\Delta_{N_h}(\lambda)
	\bigr)
	=
	1,
	$$
	where $\Delta_1(\lambda),\ldots,\Delta_{N_h}(\lambda)$ are the
	$h\times h$ minors of $P_{\mathcal L}(\lambda)$.
\end{thm}

\begin{proof}
	The preceding section characterizes all minimum augmentations satisfying
	condition $(A)$ by the parameters $(J,\Theta,K)$ appearing above.
	Therefore every augmentation of functional dimension $q_0$
	satisfying condition $(A)$ is represented by at least one such choice
	of $(J,\Theta,K)$.
	
	By Theorem~\ref{thm:reduced-condition-S}, a member of this family also
	satisfies condition $(S)$ if and only if
	$$
	\rank P_{\mathcal L}(\lambda)
	=
	\rank H,
	\qquad
	\forall\lambda\in\mathbb C.
	$$
	Hence a functional observer of order $q_0$ exists if and only
	if at least one choice of $(J,\Theta,K)$ satisfies this condition.
	
	If $h=0$, the condition holds automatically, as shown above. If
	$h>0$, the preceding maximal-minor argument shows that it is equivalent
	to \eqref{eq:gcdtest}.
\end{proof}

\begin{remark}
	Theorem~\ref{thm:minimum-A-order-existence} is stated for arbitrary
	observer-pole assignment. If only asymptotic convergence is required,
	the same result holds with the spectral test restricted to
	\(
	\rank P_{\mathcal L}(\lambda)
	=
	\rank H,
	\,\,
	\forall\lambda\in\mathbb C
	\ \text{with}\ 
	\Re(\lambda)\geq0.
	\)
	Thus an asymptotic functional observer of order $q_0$ exists
	if and only if at least one member of the minimum condition-$(A)$
	family satisfies this restricted spectral condition.
\end{remark}

\subsection{Recovery of the Darouach and Luenberger Conditions}
\label{subsec:classical-special-cases}

Theorem~\ref{thm:minimum-A-order-existence} recovers the Darouach
functional-observer conditions when $L$ requires no augmentation and
the Luenberger observability condition when $L$ spans all state
directions not directly available from the measured output.

\begin{corol}[Recovery of Darouach observer conditions]
	\label{cor:darouach-conditions}
	If $q_0=r$, then $k=0$ and
	Theorem~\ref{thm:minimum-A-order-existence} reduces to
	\(
	\rank
	\begin{pmatrix}
		LA\\ L\\ CA\\ C
	\end{pmatrix}
	=
	\rank
	\begin{pmatrix}
		L\\ CA\\ C
	\end{pmatrix},
	\)
	and
	\(
	\rank
	\begin{pmatrix}
		\lambda L-LA\\ CA\\ C
	\end{pmatrix}
	=
	\rank
	\begin{pmatrix}
		L\\ CA\\ C
	\end{pmatrix},
	\,\,
	\forall\lambda\in\mathbb C,
	\)
	which are precisely the Darouach observer-existence conditions.
\end{corol}

\begin{corol}[Recovery of Luenberger observer condition]
	\label{cor:luenberger-condition}
	Let $C\in\mathbb R^{p\times n}$ have full row rank and
	$L\in\mathbb R^{(n-p)\times n}$ satisfy
	\begin{equation}
		\rank
		\begin{pmatrix}
			C\\ L
		\end{pmatrix}
		=n.
		\label{eq:CL-nonsingular}
	\end{equation}
	Then $\eta_j=1$, $j=1,\ldots,n-p$, and
	$q_0=n-p$. Moreover,
	Theorem~\ref{thm:minimum-A-order-existence} reduces to
	\begin{equation}
		\rank
		\begin{pmatrix}
			\lambda I_n-A\\ C
		\end{pmatrix}
		=n,
		\qquad
		\forall\lambda\in\mathbb C,
		\label{eq:PBH-observability}
	\end{equation}
	the PBH observability condition for $(A,C)$.
\end{corol}

\begin{proof}
	Since $C,L$ together span $\mathbb R^{1\times n}$, condition $(A)$
	is automatic, giving $\eta_j=1$ and $q_0=n-p$; condition
	$(S)$ reduces to $\rank\Big(L(\lambda I_n-A);CA;C\Big)=n$ for all
	$\lambda$. We show this equals $\rank(\lambda I_n-A;C)$.
	
	Let $\Phi:=\begin{pmatrix}C\\L\end{pmatrix}^{-1}=(\Phi_1\ \,\Phi_2)$,
	partitioned conformably with $C,L$, so that
	$C\Phi_1=I_p$, $C\Phi_2=\mathbf0$, $L\Phi_1=\mathbf0$,
	$L\Phi_2=I_{n-p}$. Right-multiplying by $\Phi$ preserves rank and
	gives both matrices the common form
	$\begin{pmatrix}X&Y\\I_p&\mathbf0\end{pmatrix}$, with
	\[
	Y=Y_1:=(\lambda I_n-A)\Phi_2
	\quad\text{for }(\lambda I_n-A;C),
\]
\[
	Y=Y_2:=
	\begin{pmatrix}
		\lambda I_{n-p}-LA\Phi_2\\CA\Phi_2
	\end{pmatrix}
	\quad\text{for }\big((\lambda I_n-A);CA;C)\big).
	\]
	Since
	$\begin{pmatrix}I_n&-X\\ \mathbf0&I_p\end{pmatrix}
	\begin{pmatrix}X&Y\\I_p&\mathbf0\end{pmatrix}
	=\begin{pmatrix}\mathbf0&Y\\I_p&\mathbf0\end{pmatrix}$
	has rank $p+\rank(Y)$ for any $Y$, setting $Y=Y_1$ gives
	$\rank(\lambda I_n-A;C)=p+\rank(Y_1)$, and setting $Y=Y_2$ gives
	$\rank\Big(L(\lambda I_n-A);CA;C\Big)=p+\rank(Y_2)$.
	
	It remains to show $\rank(Y_1)=\rank(Y_2)$ for every $\lambda$.
	Left-multiplying $Y_1$ by the nonsingular
	$\begin{pmatrix}C\\L\end{pmatrix}$ gives
	$\begin{pmatrix}-CA\Phi_2\\ \lambda I_{n-p}-LA\Phi_2\end{pmatrix}$,
	which is $Y_2$ with its two row blocks swapped and the first
	negated. Since row-block swaps and sign changes do not affect
	rank, $\rank\bigl(\begin{pmatrix}C\\L\end{pmatrix}Y_1\bigr)
	=\rank(Y_2)$; and since
	$\begin{pmatrix}C\\L\end{pmatrix}$ is nonsingular,
	$\rank\bigl(\begin{pmatrix}C\\L\end{pmatrix}Y_1\bigr)=\rank(Y_1)$.
	Hence $\rank(Y_1)=\rank(Y_2)$, the two original ranks coincide for
	every $\lambda$, and \eqref{eq:PBH-observability} follows.
\end{proof}

\begin{remark}[Constructive decision procedure]
	The preceding results yield a finite procedure. First, apply the
	global rank-retention rule to determine
	$
	q_0=\sum_{j=1}^r\eta_j.
	$
	Construct the complete minimum condition-\((A)\) family and recover
	the corresponding augmented functional matrices as in
	Sections~\ref{subsec:completeness-minimum-A}
	and~\ref{subsec:recovery-minimum-A}. For each admissible
	representative, form
	\[
	H=\mathcal L_J(\Theta)N_{CA},
	\qquad
	F=\mathcal L_J(\Theta)AN_{CA},
	\]
	and test
	\[
	\rank(\lambda H-F)=\rank H,
	\qquad \forall\lambda\in\mathbb C.
	\]
	An observer of order \(q_0\) exists if and only if this test
	holds for at least one admissible pair \((J,\Theta)\). Once such an augmentation is found, the observer parameters can be
	constructed using the procedure in~\cite{22new}, with the functional
	matrix \(L\) therein replaced by \(\mathcal L_J(\Theta)\).
\end{remark}

\section{Example}
\label{sec:examples}

Consider the system \eqref{eq:1} with \(n=8\), \(p=r=2\), and
{\small
	\[
	A=
	\begin{pmatrix}
		-14&3&-4&5&-1&-1&1&0\\
		-3&5&3&-2&-6&0&8&0\\
		22&-17&1&-9&8&3&-11&0\\
		9&-11&-2&-6&7&1&-5&0\\
		21&-2&12&-14&-6&3&2&0\\
		0&1&-1&0&0&-5&3&0\\
		10&-9&1&-3&5&1&-10&0\\
		0&0&0&0&0&0&0&-1
	\end{pmatrix}, \]
	\[B=\begin{pmatrix}
		e_1 &e_2+e_8
	\end{pmatrix},
	\]
	\[
	C=
	\begin{pmatrix}
		2&-3&0&0&2&0&-4&0\\
		-1&-10&-3&2&5&0&-10&0
	\end{pmatrix},
	\]
	\[
	L=
	\begin{pmatrix}
		-2&3&1&0&-2&0&1&0\\
		-7&8&3&0&-6&0&2&0
	\end{pmatrix},
	\]}
where \(e_i\) is the \(i\)-th standard basis vector of \(\mathbb R^8\).
The eighth state is decoupled from, and unobserved by, the remaining
seven, so that \((A,C)\) is unobservable while the first seven rows and
columns reproduce the finite-dimensional subsystem examined below.

The functional observability indices are
$
(\eta_1,\eta_2)=(3,1),
$
and hence
\[
q_0=\eta_1+\eta_2=4.
\]
The conditions for the existence of an order-\(r\) functional observer
in \cite{22new} are not satisfied, and \((A,C)\) is unobservable, with
$
n_0=\rank\mathcal O_C=7<n=8.
$
Thus,
\[
r=2<q_0=4<n_0-p=5<n-p=6,
\]
so the minimum dimension permitted by condition~\((A)\) is strictly
between the functional dimension \(r\) and the functional-observability
endpoint \(n_0-p\), which itself lies strictly below the (here
inapplicable) reduced-order Luenberger dimension \(n-p\).

A family of minimum condition-\((A)\) augmentations is
{\small
	\[
	\mathcal L_{a,b}=
	\begin{pmatrix}
		-2&3&1&0&-2&0&1&0\\
		-7&8&3&0&-6&0&2&0\\
		-4a-1&6a+2&2a&0&-4a-1&-a&2a+2&0\\
		2-4b&6b+1&2b&-1&-4b&-b&2b+3&0
	\end{pmatrix},\]}
\[
a,b\in\mathbb R,
\]
with
$
\rank(\mathcal L_{a,b})=4.
$
Every member satisfies condition~\((A)\).

For this family, the maximal minors of the reduced spectral pencil,
up to a common nonzero factor, are
\[
\lambda+5,\quad
-(\lambda+5)^2,\quad
(\lambda+5)^3-a(\lambda+5)-b,\quad
-\bigl(a(\lambda+5)+b\bigr).
\]
Consequently,
\[
\mathcal L_{a,b}\text{ satisfies }(S)
\quad\Longleftrightarrow\quad b\neq0.
\]
In particular, the chain augmentation \(\mathcal L_{0,0}\) satisfies
condition~\((A)\) but fails condition~\((S)\) at \(\lambda=-5\), whereas
\(\mathcal L_{0,1}\) satisfies both conditions. Therefore, by
Theorem~\ref{thm:minimum-A-order-existence},
\[
\nu_{\min}=q_0=4.
\]
Thus, spectral failure of one minimum condition-\((A)\) augmentation
does not preclude another augmentation of the same dimension from
yielding a minimum-order observer.

It is also instructive to compare this result with the constructive
procedure of Darouach and Fernando~\cite{ref11n}. For \(q=4\), the
corresponding case is \(k=1\) and \(\mu=r=2\). For the present
example, condition~(43) of~\cite{ref11n} fails for every choice of
\(N\), so that procedure proceeds to the next admissible order. This
can be certified directly: the vector
\(w=(2,0,-6,-3,-6,0,-2,0)^{\mathsf T}\) satisfies
\[
\begin{pmatrix}
	C\\CA\\CA^2
\end{pmatrix}
w=0,
\qquad
\begin{pmatrix}
	L\\LA
\end{pmatrix}
w=0,
\]
yet \(LA^2w=(1,1)^{\mathsf T}\neq0\); since any row of
\(LA^2-N
\begin{pmatrix}
	L\\LA
\end{pmatrix}
\)
lying in
\(
\row
\begin{pmatrix}
	C\\CA\\CA^2
\end{pmatrix}
\)
would have to annihilate \(w\), no \(N\) can satisfy condition~(43).
The minimum-order characterization developed here complements that
construction by showing, through the complete family of minimum
condition-\((A)\) augmentations, that a spectrally feasible
fourth-order realization nevertheless exists.

To construct such a realization, we apply the functional-observer
design procedure of Darouach~\cite{22new} to
\(\mathcal L=\mathcal L_{0,1}\), with the functional matrix \(L\)
therein replaced by \(\mathcal L\). Choosing the eigenvalues of \(N\) as
$
\sigma(N)=\{-2,-3,-4,-6\},
$
yields
{\small
	\[
	\dot w(t)=
	\begin{pmatrix}
		-5&11&1&0\\
		0&0&1&0\\
		0&-5&-5&1\\
		1&-6&0&-5
	\end{pmatrix}w(t)
	+
	\begin{pmatrix}
		-50&-1\\
		-20&-1\\
		31&0\\
		19&0
	\end{pmatrix}y(t)
	\]
	\[
	+
	\begin{pmatrix}
		20&-30\\
		3&-7\\
		-12&7\\
		-14&25
	\end{pmatrix}u(t),
	\]}
with
{\small
	\[
	\widehat{\mathcal Lx(t)}
	=
	w(t)+
	\begin{pmatrix}
		-11&0\\
		-5&0\\
		5&-1\\
		6&0
	\end{pmatrix}y(t).
	\]}
Since \(N\) is Hurwitz,
$
\widehat{\mathcal Lx(t)}\longrightarrow\mathcal Lx(t).
$
Moreover, since the first two rows of \(\mathcal L\) coincide with
\(L\), the desired functional estimate is
{\small
	\[
	\hat z(t)
	=
	\begin{pmatrix}I_2&0\end{pmatrix}\widehat{\mathcal Lx(t)}
	=
	\begin{pmatrix}w_1(t)\\w_2(t)\end{pmatrix}
	+
	\begin{pmatrix}
		-11&0\\
		-5&0
	\end{pmatrix}y(t).
	\]}
Neither the Darouach nor the Luenberger observer exists here, yet the
minimum-order $q_0$ observer proposed in this paper does, as
constructed explicitly above.

\section{Conclusion}
\label{sec:conclusion}

This paper has developed a general minimum-order functional observer
framework that contains the classical Darouach functional observer and
the reduced-order Luenberger observer as special cases. The minimum
functional dimension permitted by the algebraic condition $(A)$ was
characterized by the functional observability indices, and the complete
family of augmentations attaining this dimension was determined.
Necessary-and-sufficient spectral feasibility conditions were then
established for determining whether the minimum observer order is
$\nu_{\min}=q_0=\sum_{j=1}^{r}\eta_j$; considering the complete
family is essential, since one minimum condition-$(A)$ augmentation may
fail condition $(S)$ while another augmentation of the same dimension
satisfies it.

The general result recovers the Darouach observer and its
necessary-and-sufficient existence conditions when $q_0=r$, and
the reduced-order Luenberger observer and its classical observability
condition when $q_0=n-p$. More generally, when
$r<q_0<n-p$ and spectral feasibility holds, the framework
yields a minimum-order functional observer of intermediate order that
may exist even when neither classical observer is available.

\section*{References}

	\section*{Biography}

\begin{wrapfigure}{l}{1in}
	\includegraphics[width=1in,height=1.25in,clip,keepaspectratio]{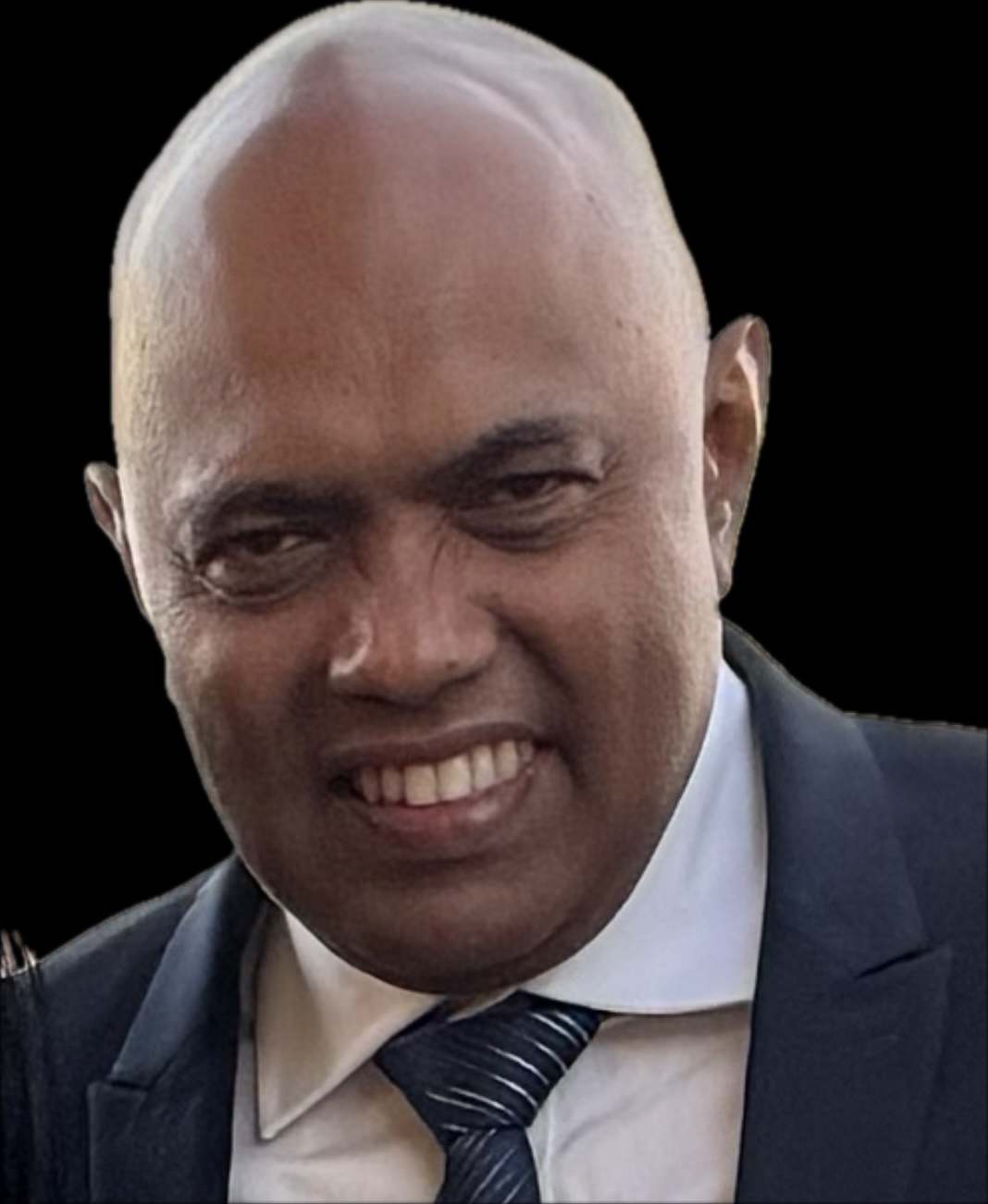}
\end{wrapfigure}
\noindent
\textbf{Tyrone Fernando} received his B.E. (Hons.) and Ph.D. degrees in Electrical Engineering from the University of Melbourne, Victoria, Australia, in 1990 and 1996, respectively. In 1996, he joined the Department of Electrical, Electronic and Computer Engineering, University of Western Australia (UWA), Crawley, WA, Australia, where he currently holds the position of Professor of Electrical Engineering. He previously served as Associate Head and Deputy Head of the department from 2008 to 2010.

Prof. Fernando is currently the Head of the Power and Clean Energy Research Group at UWA. He has provided professional consultancy to Western Power on the integration and management of distributed energy resources in the electric grid. In recognition of his professional contributions, he was named the Outstanding WA IEEE PES/PELS Engineer in 2018. His research interests include theoretical control, observer design, and power system stability and control. He has received multiple teaching awards from UWA in recognition of his contributions to control systems and power systems education.

	\end{document}